\documentclass[11pt,letterpaper]{article}

\usepackage{
amsfonts,
amsmath,
amssymb,
amsthm,
color,
float,
graphicx,
latexsym,
mathtools, 
microtype,
multirow
}

\usepackage{
boxedminipage,
etoolbox,
framed,
ifthen,
tabularx,
tcolorbox,
thm-restate,
tikz,
xifthen
}

\usepackage{fullpage}

\usepackage{cleveref}

\usetikzlibrary{calc}

\RequirePackage[T1]{fontenc}
\RequirePackage[utf8]{inputenc}

\usepackage[ruled,linesnumbered]{algorithm2e}

\theoremstyle{plain}
\newtheorem{reduction}{Reduction rule}
\newtheorem{theorem}{Theorem}
\newtheorem{observation}{Observation} 
\newtheorem{corollary}{Corollary} 
\newtheorem{lemma}{Lemma} 
\newtheorem{claim}{Claim} 
\newtheorem{proposition}{Proposition}

\newenvironment{cproof}{\proof[Proof of claim]}{\endproof}

\definecolor{black}{rgb}{0, 0, 0}

\definecolor{white}{rgb}{1, 1, 1}

\definecolor{grey}{rgb}{.6, .6, .6}

\definecolor{red}{rgb}{1, 0 ,0}

\definecolor{green}{rgb}{0, 1, 0}

\definecolor{blue}{rgb}{0, 0 ,1}

\definecolor{darkred}{rgb}{0.7, 0 ,0}

\definecolor{darkgreen}{rgb}{0, 0.7, 0}

\definecolor{darkblue}{rgb}{0, 0 , 0.55}

\definecolor{magenta}{rgb}{1, 0, 1}

\definecolor{cyan}{rgb}{0, 1, 1}

\definecolor{yellow}{rgb}{1, 0.9, 0}

\definecolor{purple}{rgb}{0.5, 0, 0.5}

\definecolor{orange}{rgb}{1, 0.5, 0}

\definecolor{turquoise}{rgb}{0, 0.7, 0.7}

\newcommand{\Ccal}{\mathcal{C}}

\newcommand{\Hcal}{\mathcal{H}}

\newcommand{\Ocal}{\mathcal{O}}

\newcommand{\Scal}{\mathcal{S}}

\newcommand{\Ucal}{\mathcal{U}}

\newcommand{\Wcal}{\mathcal{W}}

\newcommand{\bN}{\mathbb{N}}

\newcommand{\diam}{{\sf diam}}
\newcommand{\FPT}{{\sf FPT}\xspace}

\newcommand{\NP}{{\sf NP}\xspace}
\newcommand{\coNP}{{\sf co-NP}\xspace}
\newcommand{\NPp}{{\sf NP/poly}\xspace}
\newcommand{\no}{{\sf no}\xspace}
\newcommand{\para}{{\sf para-NP}\xspace}

\newcommand{\XP}{{\sf XP}\xspace}
\newcommand{\yes}{{\sf yes}\xspace}
\newcommand{\W}[1]{{\sf W[#1]}\xspace}

\newlength{\RoundedBoxWidth}
\newsavebox{\GrayRoundedBox}
\newenvironment{GrayBox}[1]%
   {\setlength{\RoundedBoxWidth}{.93\textwidth}
    \def\boxheading{#1}
    \begin{lrbox}{\GrayRoundedBox}
       \begin{minipage}{\RoundedBoxWidth}}%
   {   \end{minipage}
    \end{lrbox}
    \begin{center}
    \begin{tikzpicture}%
       \node(Text)[draw=black!20,fill=white,rounded corners,%
             inner sep=2ex,text width=\RoundedBoxWidth]%
             {\usebox{\GrayRoundedBox}};
        \coordinate(x) at (current bounding box.north west);
        \node [draw=white,rectangle,inner sep=3pt,anchor=north west,fill=white] 
        at ($(x)+(6pt,.75em)$) {\boxheading};
    \end{tikzpicture}
    \end{center}}
\newenvironment{defproblemx}[2][]{\noindent\ignorespaces%
                                \FrameSep=6pt%
                                \parindent=0pt%
                \vspace*{-1.5em}
                \ifthenelse{\isempty{#1}}{%
                  \begin{GrayBox}{\textsc{#2}}%
                }{%
                  \begin{GrayBox}{\textsc{#2}  parameterized by~{#1}}%
                }
                \begin{tabular*}{\textwidth}{@{\hspace{.1em}} >{\itshape} p{1.8cm} p{0.8\textwidth} @{}}%
            }{
                \end{tabular*}%
                \end{GrayBox}%
                \ignorespacesafterend
            }

\usepackage[tickmarkheight=0.1cm]{todonotes}

\newcommand{\id}{{\sc Identification}\xspace}
\newcommand{\Hid}[1]{{\sc #1-Identification}\xspace}
\newcommand{\kidH}[1]{{\sc #1-Identification}\xspace}

\title{Identification to Subclasses of Chordal Graphs\thanks{A preliminary version of this paper appeared in the proceedings of IPEC 2026. The research has been supported by the Research Council of Norway via BWCA Grant 314528 and Extreme-Algorithms Grant 355137, and
by the Leverhulme Trust Grant RPG-2024-182.}}

\author{Petr A. Golovach\thanks{Department of Informatics, University of Bergen, Bergen,                   Norway}\\petr.golovach@uib.no
		\and
		Laure Morelle\addtocounter{footnote}{-1}\footnotemark{}\\laure.morelle@uib.no
		\and 
        Dani\"el Paulusma\thanks{Department of Computer Science, Durham University, Durham, United Kingdom}\\daniel.paulusma@durham.ac.uk
	}

\date{}

\begin{document}

\maketitle

\begin{abstract}
An identification of two vertices $u$ and $v$ in a graph replaces them with a new vertex that is made adjacent to every vertex of $N(u)\cup N(v)\setminus \{u,v
\}$.  
We study the {\sc ${\cal H}$-Identification} problem, which is to decide whether a given graph $G$ can be transformed (``identified'') to a graph in ${\cal H}$ by applying at most $k$ vertex identifications. We determine the classical and parameterized complexity of this problem for various subclasses ${\cal H}$ of chordal graphs, obtaining an almost complete picture for two parameters: $k$ and $n-k$. We also consider the {\sc Identification} problem, which is to test, for two given graphs $G$ and $H$, whether $G$ can be identified to $H$. We determine the parameterized complexity of this problem when $H$ is a graph from one of our testbed classes, taking the number of simplicial vertices of $H$ as the parameter. 

\end{abstract}

\section{Introduction}\label{sec:intro}

An {\it identification} of two vertices $u$ and $v$ in a graph $G$ replaces $u$ and $v$ by a new vertex~$w_{uv}$ that is made adjacent to every vertex of $G-\{u,v\}$ that is a neighbor of $u$ or $v$ in $G$. Equivalently, one may first add the edge $uv$ if it is not already present and then contract this edge. For a fixed class of graphs~${\cal H}$, we define the problem {\sc ${\cal H}$-Identification}. The input consists of a graph $G$ and an integer~$k$, and the question is whether $G$ can be transformed into some graph $H\in {\cal H}$ using at most $k$ identifications. We study the parameterized complexity of {\sc ${\cal H}$-Identification} and provide a comprehensive complexity analysis for well-known subclasses~${\cal H}$ of chordal graphs.

The {\sc ${\cal H}$-Identification} problem falls within the framework of graph modification problems, where the goal is to decide whether a given graph can be modified into a graph in a class~${\cal H}$ using at most $k$ operations from some specific set $S$ of allowed modifications. For example, if $S$ consists of vertex deletion, edge deletion, and edge contraction, then we obtain the ${\cal H}$-{\sc Minor} problem. Another well-known variant related to \Hid{$\cal H$} is {\sc ${\cal H}$-Contraction}, for which $S$ consists solely of edge contraction. 

Most graph modification problems are \NP-complete and have been widely studied in parameterized complexity. For example, if ${\cal H}$ is minor-closed, such as forests, series-parallel graphs, planar graphs, and graphs of bounded treewidth, then, for fixed~$k$, the yes-instances of
${\cal H}$-{\sc Minor} form a minor-closed graph class as well. 
The same holds
for \Hid{$\cal H$}: 
if $G$ can be turned into some $G^* \in \mathcal{H}$ by at most $k$ identifications, then any minor $G'$ of $G$ can mimic these identifications, yielding a minor of $G^*$. Since $\mathcal{H}$ is minor-closed, $G'$ is also a yes-instance.
Hence, 
both problems can be solved in 
$f(k,{\cal H})\cdot n^{1+o(1)}$
time~\cite{KorhonenPS24mino,RobertsonS04XX} (see also~\cite{GorskySW26thep,MorelleST25grap}).  
This argument does not hold for ${\cal H}$-{\sc Contraction}~\cite{MorelleST25grap} (see e.g.~\cite{GHP13,GM13,HHLLP14,HHLP13}). 

If the graphs in ${\cal H}$ have treewidth at most~$p$ and can be characterized in CMSO logic, such as forests and outerplanar graphs, then for the class of graphs in the \yes-instances of \Hid{$\cal H$}, the treewidth is upper bounded by $p+k$, and these graphs also admit a CMSO characterization when $k$ is fixed.  
Hence, Courcelle's Theorem~\cite{ArnborgLS91easy,Courcelle90them} yields a $g(p,k)\cdot n$ time algorithm.

However, when ${\cal H}$ is not minor-closed or does not have bounded treewidth, not much is known except for the case when ${\cal H}=\{H\}$ for some single graph~$H$. In this case,
we simply write $H$-\id and note that 
since each identification reduces the number of vertices by one, the parameter $k$ is fixed by the instance and can be omitted.\footnote{The same holds for $H$-{\sc Contraction},
but not for $H$-{\sc Minor}~\cite{GPS17}. While the latter is polynomial-time solvable for every fixed graph~$H$ when no bound on $k$ is imposed~\cite{RS95}, both $H$-{\sc Contraction} and $H$-{\sc Minor} are \NP-complete for some graphs~$H$~\cite{BV87,GPS17} and are far from being classified (see~\cite{GPS17,LPW08,LPW08b,HKPST12}).} 
To explain the known results, we first introduce an equivalent formulation. Two disjoint vertex subsets $S$ and $T$ are {\it adjacent} if there is an edge with one endpoint in $S$ and the other in $T$. A graph $G$ can be identified to a graph $H$ if and only if $V(G)$ can be partitioned into {\it bags} $\{W(x)\mid x\in V(H)\}$ such that

\begin{itemize}
\item [(i)] for every $x\in V(H)$, $W(x)$ is non-empty; and
\item [(ii)] for every pair of distinct vertices $x,y\in V(H)$, $W(x)$ and $W(y)$ are adjacent if and only if $xy\in E(H)$.
\end{itemize}

\noindent
By identifying the vertices within each $W(x)$ to a single vertex, we obtain the graph $H$. 
We call ${\cal W}=\{W(x)\mid x\in V(H)\}$ an {\it $H$-witness structure} of $G$. See \Cref{fig:ex_id} for an illustration.
\begin{figure}[t]
    \centering
    \includegraphics[scale=1]{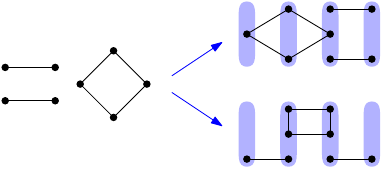}
    \caption{An example of a graph $G$ that can be identified to $P_4$, as shown with two different $P_4$-witness structures. Note that $G$ is disconnected and cannot be contracted to $P_4$.}
    \label{fig:ex_id}
\end{figure}

For graphs $G$ and $H$, a function $f:V(G)\to V(H)$ is called a {\it homomorphism} from $G$ to $H$ if $f(u)f(v)\in E(H)$ for every edge $uv\in E(G)$. A homomorphism $f:V(G)\to V(H)$ is a {\it compaction} if, for each vertex $x\in V(H)$, there exists a vertex $v\in V(G)$ such that $f(v)=x$ and, for each edge $xy\in E(H)$ with $x\neq y$, there exists an edge $uv\in E(G)$ such that $f(u)=x$ and $f(v)=y$. Here, we write $x\neq y$ to allow for the possibility that $E(H)$ contains self-loops $xx$. A graph is {\it reflexive} if all its vertices have a self-loop. For an {\it irreflexive} graph~$H$ (i.e. a graph with no self-loops), let $H^*$ be the reflexive graph obtained from $H$ by adding a self-loop to every vertex of $H$. Now, by using (i) and (ii), it follows that $G$ can be identified to $H$ if and only if there is a compaction from $G$ to $H^*$. Thus, for every irreflexive graph $H$, the {\sc $H$-Identification} problem can be reformulated as the {\sc $H^*$-Compaction} problem, which asks whether a given (irreflexive) graph~$G$ admits a compaction to the fixed (reflexive) graph~$H^*$.  

Classifying the complexity of {\sc $H^*$-Compaction} is still open; see also the survey~\cite{BKM12}. However, Vikas
showed that {\sc $H^*$-Compaction}, and thus {\sc $H$-Identification}, is \NP-complete if $H=C_r$ for any $r\geq 4$~\cite{Vikas02comp}, but in
\XP\ when parameterized by $|V(H)|$
if $H$ is chordal~\cite{Vikas04comp}; a graph is {\it chordal} if it has no induced cycle on four or more vertices. 
This result for $H$-\id 
is tight in the sense that even $K_{1,\ell}$-\id is \W{1}-hard when parameterized by $\ell$, where $K_{1,\ell}$ denotes the star on $\ell+1$ vertices (which is chordal).
Indeed, by (i) and (ii), a graph $G$ without isolated vertices can be identified to~$K_{1,\ell}$ if and only if $G$ has an independent set of size at least~$\ell$, and the {\sc Independent Set} problem is \W{1}-complete when parameterized by solution size~\cite{DF95}. 
These results lead to our main research~question:

\smallskip
\noindent
{\it What is the parameterized complexity of {\sc ${\cal H}$-Identification} beyond the single-graph classes ${\cal H}=\{H\}$, when ${\cal H}$ ranges over \emph{subclasses} of chordal graphs?}

\smallskip
\noindent
We concentrate on the parameterizations by $k$ and $n-k$. In addition, we seek tight algorithmic results for {\sc $H$-Identification} when $H$ is chordal; note that in this case, the parameterization by $n-k$ is the same as the parameterization by $|V(H)|$. 
Note also that both the input graph~$G$ and the fixed graph $H$ in {\sc $H$-Identification} are irreflexive. 

\medskip
\noindent
{\bf Our Results.}
We use the following subclasses of chordal graphs as testbeds for our study: paths, linear forests, stars, trees, forests, complete graphs, clusters, split graphs, interval graphs and chordal graphs. Note that forests have treewidth at most~$1$, implying that \Hid{$\cal H$} can be solved in $g(k)\cdot n$ time via meta-theorems~\cite{ArnborgLS91easy,Courcelle90them} for most subclasses~${\cal H}$ of forests, as observed earlier. In these cases, we still aim for a stronger result. All our other testbed classes ${\cal H}$ have unbounded treewidth and are not minor-closed, and therefore, none of the previous meta-theorems~\cite{ArnborgLS91easy,Courcelle90them,KorhonenPS24mino,MorelleST25grap,RobertsonS04XX} apply.

As our first result, we prove that for all our testbed classes~${\cal H}$,  {\sc ${\cal H}$-Identification} is \NP-complete, except for the cases where ${\cal H}$ is  the class of paths or linear forests. Specifically, when ${\cal H}$ is the class of paths, the problem is known to be solvable in $O(nm)$ time~\cite{Vikas13algo}. We show that the proof of this result can be extended to hold even for linear forests.
 
We also observe that for each testbed class~${\cal H}$, it is possible to check in polynomial time whether a given graph $G$ belongs to ${\cal H}$. Hence, 
for every such class ${\cal H}$,
\Hid{$\cal H$} is in \XP\ when parameterized by $k$. Namely, we can consider
all $\sum_{i=0}^k\binom{N}{i}\le n^{2k}$, where $N=\binom{n}{2}$, sets~$T$ of at most $k$ vertex pairs and check in polynomial time for each $T$, whether identifying the pairs in $T$ yields a graph in ${\cal H}$. 
We can improve this \XP\ result to a polynomial kernel for stars, trees and forests, and to an \FPT\ algorithm for
complete graphs, cluster graphs and split graphs. See 
\Cref{t-1}, in which we also show that these results are best possible for our testbed classes except for two open cases: we do not know whether {\sc ${\cal H}$-Identification}, parameterized by~$k$, is \FPT\ when ${\cal H}$ is the class of interval graphs, nor when ${\cal H}$ is the class of chordal graphs; though we show that these cases do not admit a polynomial kernel (unless \coNP$\subseteq$~\NPp).

We now argue that $n-k$ is another natural choice of parameter that immediately leads to \XP\ algorithms for all our testbed classes. To see this, we first observe that all our classes are closed under edge contraction. Therefore, we can perform exactly $k$ identifications instead of at most 
$k$ whenever the obtained graph is not edgeless, since any additional edge contractions can be applied without leaving the class. For classes containing edgeless graphs with at least two vertices, we observe that they are hereditary, and the identification of two isolated vertices 
is equivalent to deleting one of them. 
As a result, we may assume that any obtained solution will have $n-k$ vertices. Now consider the set ${\cal G}_{n-k}$ of all graphs on exactly $n-k$ vertices. For each $H\in {\cal G}_{n-k}$, we check in polynomial time if $H\in {\cal H}$, which would imply in particular that $H$ is also chordal. Thus, we can apply the algorithm of Vikas~\cite{Vikas04comp} to check if $G$ can be identified to $H$. Recall that the algorithm of~\cite{Vikas04comp} runs in \XP\ time when parameterized by $|V(H)|$. Since for every $H\in {\cal G}_{n-k}$, $|V(H)|=n-k$, and the size of ${\cal G}_{n-k}$ is bounded by $2^{(n-k)^2}$, the result follows.

For parameter $n-k$, we were able to obtain a complete picture, as shown in \Cref{t-1}, which summarizes all our results for this parameter as well.  
From \Cref{t-1}, we can see in particular that for sparse chordal graph classes, positive results are obtained
for parameter~$k$, while for dense chordal graph classes this is the case for parameter $n-k$.  

Finally, we consider the
\id problem, which has two graphs $G$ and~$H$ as input and asks whether $G$ can be identified to $H$. Every chordal graph on two or more vertices has at least two {\it simplicial} vertices (vertices whose neighborhood is a clique)~\cite{Di61}. Hence, in this setting, the number of simplicial vertices becomes a natural parameter. As we will explain, the result of Vikas~\cite{Vikas13algo} immediately gives polynomial-time solvability if $H$ is a path. However, if $H$ is a linear forest, then the problem becomes \W{1}-hard. This is shown in \Cref{t-2}, in which we also summarize our other results for this parameter.

\begin{table}[t]
\centering
\begin{tabular}{|c|c|c|c|}
\hline
$\Hcal $&\NP-c & parameter: $k$ & parameter: $n-k$\\
\hline
paths~\cite{Vikas13algo}, linear forests & no 
& {\sf P} 
& {\sf P}\\ 
\hline
stars, trees, forests &yes 
& poly-kernel 
&\XP\& \W{1}-hard\\ 
\hline
complete, cluster, split graphs  &yes & \FPT\ 
\& no poly-kernel 
&poly-kernel\\
\hline
interval, chordal graphs &yes & \XP \& no poly-kernel & poly-kernel\\
\hline
\end{tabular}
\caption{Our results and the result from~\cite{Vikas13algo} for
{\sc ${\cal H}$-Identification} for
chordal subclasses.}\label{t-1} 

\medskip
\centering
\begin{tabular}{|c|c|}
\hline
$H$ & parameter: $\ell$ (number of simplicial vertices of $H$)\\
\hline
path~\cite{Vikas13algo} & {\sf P} \\
\hline
linear forest, 
star, tree, forest 
&\XP\ \& \W{1}-hard\\
\hline
complete graph, cluster graph  &poly-kernel\\
\hline
split graph, interval graph, chordal graph & 
para-\NP-complete (for $\ell=2$) \\
\hline
\end{tabular}
\caption{Our results and the result of~\cite{Vikas13algo} for \id for instances $(G,H)$ when $H$ belongs to a specific subclass of chordal graphs.}\label{t-2}
\end{table}

\medskip
\noindent
{\bf Methods.} 
Our results are closely linked. In many cases, we use an algorithm developed for one result as a subroutine in another, allowing us to extend our techniques. For example, we first develop an \FPT\ algorithm for \Hid{$\cal H$}, parameterized by~$k$, when ${\cal H}$ is the class of complete graphs. We then use this algorithm as a black box to obtain \FPT\ algorithms for the cases where ${\cal H}$ is the class of cluster graphs and the class of split graphs. Moreover, it is not possible to adapt known techniques from other graph modification problems due to the nature of the identification operation. This is illustrated in \Cref{fig:ex_id}, which shows a disconnected graph $G$ that can be identified to the connected graph $P_4$. As a consequence, the problem becomes significantly more challenging when $G$ is disconnected and requires additional combinatorial arguments. This difficulty is, for example, reflected in our \XP\ algorithm for \id on pairs $(G,H)$ where $H$ is a forest.

\medskip
\noindent
{\bf Paper Organization.}
In \Cref{sec:prelim} we introduce the basic notions used throughout the paper and provide some preliminary observations. Our results are split into three sections according to the graph class considered:
the results of Tables~\ref{t-1} and~\ref{t-2} on acyclic graphs are in \Cref{sec:acyclic}; on complete and cluster graphs in \Cref{sec:cliques}; on split, interval and chordal graphs in \Cref{sec:chord}.
We conclude our paper in \Cref{sec:ccl} by discussing further research directions.

\section{Preliminaries}\label{sec:prelim}
For $p\geq 1$, we write $[p]=\{1,\ldots,p\}$, and for $p\leq q$, we write $[p,q]=\{p,\ldots,q\}$.
Let $G=(V(G),E(G))$ be a graph. We denote an edge $\{x,y\}$ by
$xy$ (or $yx$). If no confusion is possible, we may also write $V(G)=V$ and $E(G)=E$. We usually let $n$ and $m$ denote the number of vertices and edges, respectively, if $G$ is part of the input. 
We write $G[U]$ to denote the subgraph of $G$ \emph{induced} by $U\subseteq V(G)$.
Given $U\subseteq V$, we also write $G-U$ to denote the graph obtained after removing the vertices of $U$. If $U=\{u\}$, we write $G-u$.
For a set of edges $R$, $G[R]$ denotes the subgraph of $G$ whose vertices are the endpoints of the edges of $R$ and whose set of edges is $R$.
The \emph{neighborhood} $N_G(v)$ of a vertex $v\in V(G)$ is the set of all vertices \emph{adjacent} to $v$, i.e. all vertices $u$ such that $uv\in E(G)$.
We also write $N_G[v]=N_G(v)\cup\{v\}$.
For a set $U\subseteq V(G)$, $N_G[U]=\bigcup_{v\in U}N_G[v]$ and $N_G(U)=N_G[U]\setminus U$.
A set of vertices $U$ is a \emph{vertex cover} of $G$ if for every edge $uv\in E(G)$, $u\in U$ or $v\in U$. 

A \emph{separation} of a graph $G$ is a pair $(A,B)$ of subsets of $V(G)$ such that $A\cup B=V(G)$ and there is no edge $uv$ with $u\in A\setminus B$ and $v\in B\setminus A$; $|A\cap B|$ is the \emph{order} of the separation.
The {\it length} of a path $P$ is its number of edges.
If $P$ has end-vertices $u$ and $v$, we say that $P$ is a
\emph{$u$-$v$-path}.
The \emph{distance} between vertices $u$ and $v$ is the smallest length of a $u$-$v$-path. The \emph{diameter} $\diam(G)$ of a connected graph $G$ is the maximum distance between two vertices of $G$.

A graph $G$ contains a graph $H$ as an {\it identification} (or $G$ can be {\it identified} to $H$) if $G$ can be transformed to $H$ by a sequence of identifications.
Recall that in such a case, $G$ has an $H$-witness structure ${\cal W}$.
If a bag of $\Wcal$ contains only one vertex, then we say that it is \emph{trivial}.

For algorithmic results, we use the following useful property of bags in witness structures corresponding to vertices of degree~$1$.

\begin{lemma}\label{lem:singletons}
Let $G$ and $H$ be graphs, and let $U\subseteq V(H)$ be a set of pairwise non-adjacent vertices of degree~$1$. If $G$ 
can be identified to $H$, then there is an $H$-witness structure
$\Wcal=\{W(x)\mid x\in V(H)\}$ such that $|W(x)|=1$ for all $x\in U$.
\end{lemma}

\begin{proof}
Let $\Wcal=\{W(x)\mid x\in V(H)\}$ be an $H$-witness structure such that the total number of vertices in $\bigcup_{x\in U}W(x)$ is minimum. We claim that $|W(x)|=1$ for all $x\in U$. 
For the sake of contradiction, assume that there is $y\in U$ such that $|W(y)|\geq 2$. Let $z$ be the unique neighbor of $y$ in $H$; note that $z\notin U$. Since $\Wcal$ is an $H$-witness structure, there is $v\in W(y)$ that has a neighbor in $W(z)$. Consider  
$\Wcal'=\{W'(x)\mid x\in V(H)\}$, where $W'(y)=\{v\}$, $W'(z)=W(z)\cup (W(y)\setminus \{v\})$, and $W'(x)=W(x)$ for all $x\in V(H)$ distinct from $y$ and $z$. As $v\in W'(y)$ is adjacent to a vertex of $W(z)\subseteq W'(z)$ and $N_G[W(y)]\subseteq W(y)\cup W(z)$ because $y$ is a pendant vertex, we obtain that $\Wcal'$ is an $H$-witness structure. However, the total number of vertices in $\bigcup_{x\in U}W'(x)$ is less than the size of $\bigcup_{x\in U}W(x)$
contradicting the choice of $\Wcal$. This completes the proof.
\end{proof}

We let $P_n$ denote the path on $n$ vertices. A forest is \emph{linear} if every connected component is a path. A vertex of degree at most~$1$ in a forest is called a \emph{leaf}, and other vertices are \emph{internal}. Note that isolated vertices are considered to be leaves.
A \emph{star} is a tree whose vertices, except at most one (called \emph{center}), are leaves.
A graph whose vertex set is a clique is called \emph{complete}. We let $K_n$ denote the complete graph on $n$ vertices. 
Slightly abusing notation, we may call a complete graph a clique.
A \emph{cluster graph} is a graph whose connected components are complete graphs.
A \emph{split graph} is a graph $G$ whose vertex set can be partitioned into a clique $K$ and an independent set $I$. Notice that such a partition is not always unique, and one of the sets may be empty. However, we always assume that $K\neq\emptyset$.
An \emph{interval graph} is a graph isomorphic to the intersection graph of intervals of the real line. 
Note that forests, complete graphs, split graphs, and interval graphs have no induced cycles on four or more vertices and thus are all chordal.

We will need an observation and a known result to obtain~\Cref{cor:paths}.

\begin{observation}\label{obs:connect}
Any identification of a connected graph is connected.
\end{observation}

\begin{proposition}[\cite{Vikas13algo}]\label{prop:path}
Let $G$ be a graph with connected components $G_1,\ldots,G_s$. Then $G$ can be identified to a path $P_\ell$ with $\ell$ vertices if and only if $\sum_{i=1}^s\diam(G_i)\geq \ell-1$.  
\end{proposition}

\begin{corollary}\label{cor:paths}
\Hid{$\cal H$} is solvable in $\Ocal(nm)$ time if $\Hcal$ consists of paths or linear forests.
\end{corollary}

\begin{proof}
If $\Hcal$ is the class of paths, then \Cref{prop:path} implies that $(G,k)$ is a \yes-instance of the problem if and only if $G$ 
can be identified, with at most $k$ identifications, to $P_\ell$ for $\ell=\sum_{i=1}^s\diam(G_i)+1$, where $G_1,\ldots,G_s$ are connected components of $G$. The latter condition holds if and only if $k\geq n-\ell$. Because the connected components can be found in $\Ocal(n+m)$ and their diameters can be computed in time $\Ocal(nm)$ by the standard breadth-first search algorithm, the claim follows.
If $\Hcal$ is the class of linear forests, then by \Cref{prop:path},
$(G,k)$, where $G$ has connected components $G_1,\ldots,G_s$, is a \yes-instance if and only if $G$ can be identified with at most $k$ identifications to the linear forest 
with connected components $P_{\ell_1},\ldots,P_{\ell_s}$, where $\ell_i=\diam(G_i)+1$ for $i\in[s]$. This concludes the proof. 
\end{proof}

\section{Identification to Acyclic Graphs}\label{sec:acyclic}
In this section, we consider \Hid{$\Hcal$} for the cases when $\Hcal$ is a class of stars, trees, or forests. We refer to those problems as \kidH{Star}, \kidH{Tree}, and \kidH{Forest}, respectively.

\subsection{Algorithmic Lower Bounds}
We start with establishing hardness results for  \Hid{$\Hcal$} where $\Hcal$ is a class of acyclic graphs.

\begin{theorem}\label{th:NPtree}
Let $\Hcal$ be the class of stars, trees, or forests.
Then \Hid{$\Hcal$} is \NP-complete and \W{1}-hard parameterized by $n-k$.
\end{theorem}

\begin{proof}
We reduce from {\sc Independent Set}. The \textsc{Independent Set} problem asks, given a graph $G$ and a positive integer $p$, whether $G$ has an independent set of size at least $p$. 
\textsc{Independent Set} is well-known to be \NP-complete~\cite{GareyJ79} and \W{1}-complete when parameterized by the solution size $p$~\cite{DowneyF13}.
Let $(G,p)$ be an instance of {\sc Independent Set} for $p\geq 2$.
Denote by  $G'$ the graph obtained from $G$ by adding a universal vertex $u$ (cf. \Cref{fig:NPtree}). We claim that $G$ has an independent set $I$ of size $p$ if and only if 
$(G',|V(G')|-(p+1))$ is a \yes-instance of \Hid{$\Hcal$}.

\begin{figure}[ht]
\center
\includegraphics[scale=1]{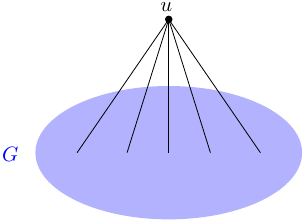}
\caption{The graph $G'$ constructed in the proof of \Cref{th:NPtree}.}
\label{fig:NPtree}
\end{figure}

Suppose that $(G,p)$ is a \yes-instance of {\sc Independent Set}.
Let $I$ be an independent set in $G$ of size~$p$.
Then, by identifying all vertices in $G'-I$ to a single vertex $u'$, we obtain a star on $p+1$ vertices.
So, $(G',|V(G')|-(p+1))$ is a \yes-instance of \kidH{$\Hcal$} when $\Hcal$ is the class of stars, trees, or forests.

Suppose now that $(G',|V(G')|-(p+1))$ is a \yes-instance of
\Hid{Forest}. 
Let $F$ be a forest obtained by doing at most $|V(G')|-(p+1)$ identifications. Clearly, $F$ has at least $p+1$ vertices. Given that $G'$ is a connected graph of diameter at most~$2$, so is $F$ by~\Cref{obs:connect} and
the fact that the diameter of a graph cannot increase after an identification. Thus, $F$ is a star with at least $p$ leaves.
Thus, $(G',|V(G')|-(p+1))$ is a \yes-instance of \kidH{Forest} if and only if it is a \yes-instance 
of \kidH{Star} (\kidH{Tree}).  
Let $\Wcal$ be an $F$-witness structure of $G'$.
Then, by picking one vertex from the bag of each leaf, we obtain an independent set $I$ of size at least $p$. Note that because $u$ is universal, it does not belong to any independent set of size at least $p\geq 2$. Therefore, $I$ is an independent set in $G$.  
So, $(G,p)$ is a \yes-instance of {\sc Independent Set}. This concludes the proof.
\end{proof}

Since a star with $p$ vertices is unique and has $p-1$ leaves when $p\geq 3$, \Cref{th:NPtree} immediately implies the following algorithmic lower bound where the target is given as a part of the input.

\begin{corollary}\label{cor:given_forest}   
\id is \W{1}-hard on instances $(G,F)$ where $F$ is a forest, when parameterized by the number of leaves of $F$.
\end{corollary}

We also observe that, while \id can be solved in $\Ocal(nm)$ time on instances $(G,P)$ where $P$ is a path by~\Cref{prop:path}, the problem becomes hard when the target graph is a linear forest.

\begin{theorem} \label{thm:lin_forest_hard}
\id is \W{1}-hard on instances $(G,L)$ where $L$ is a linear forest, when parameterized by the number of connected components of $L$.  
\end{theorem}

\begin{proof}
We reduce from the \textsc{Bin Packing} problem. In this problem, the input contains a set 
of $n$ items of positive integer sizes $s_1,\ldots,s_n$ and positive integers $k$ and $B$, and the task is to decide whether there is a partition $I_1,\ldots,I_k$ of $[n]$ such that
for every $i\in[k]$, $\sum_{j\in I_i}s_j\leq B$. \textsc{Bin Packing} parameterized by $k$ is known to be \W{1}-hard when item sizes are encoded in unary~\cite{JansenKMS13}. Furthermore, the result holds when $\sum_{i=1}^ns_i=kB$, that is, the capacity of each bin should be fully used. We consider such an instance of \textsc{Bin Packing} and construct the graphs $G$ and $L$:
\begin{itemize}
\item $G$ is the disjoint union of $n$ paths $P_1,\ldots,P_n$ of length $s_1,\ldots,s_n$, respectively,
\item $L$ is the disjoint union of $k$ paths $Q_1,\ldots,Q_k$ of length $B$.
\end{itemize}
We claim that $G$ can be identified to $L$ if and only if the considered instance of \textsc{Bin Packing} is a \yes-instance.

Suppose that there is a partition $I_1,\ldots,I_k$ of $[n]$ such that
for every $i\in[k]$, $\sum_{j\in I_i}s_j=B$. For each $i\in[k]$, we consider the paths $P_j$ with $j\in I_i$ and successively identify their endpoints to obtain a single path. Note that this is a path of length $B$. Thus, $L$ is an identification of $G$. For the opposite direction, assume that $G$ can be identified to $L$. For each $i\in[k]$, let $I_i$ be the set of indices of the connected components of $G$ whose vertices are used to obtain $Q_i$. Notice that by~\Cref{obs:connect}, the vertices of the same connected component cannot be used in distinct $Q_i$ and $Q_j$. By \Cref{prop:path}, the total length of the paths $P_j$ with $j\in I_i$ is at least $B$ for each $i\in[k]$. Then, because the total length of paths $P_1,\ldots,P_n$ is the same as the total length of $Q_1,\ldots,Q_k$, we obtain that the total length of the paths $P_j$ with $j\in I_i$ is exactly $B$ for each $i\in[k]$. Thus, for every $i\in[k]$, $\sum_{j\in I_i}s_j=B$, and the partition $I_1,\ldots,I_k$ is a solution to the instance of \textsc{Bin Packing}. This concludes the proof.
\end{proof}

\subsection{Kernelization for Identification to Acyclic Graphs}
In this section, we investigate the parameterized complexity of \Hid{$\Hcal$} parameterized by $k$ for the cases when $\Hcal$ is the class of stars, trees, or forests. 
We show that in all these cases, \Hid{$\Hcal$} admits a polynomial kernel when parameterized by $k$.\medskip

For stars, we observe that the problem is equivalent to \textsc{Vertex Cover}.  
We remind the reader that the task of the \textsc{Vertex Cover} problem is to decide, given a graph $G$ and a non-negative integer $k$, whether $G$ has a vertex cover of size at most $k$.

\begin{observation}\label{obs:vc}
Let $G$ be a graph 
with $p$ isolated vertices, and let $k\geq 0$ be an integer. Then $G$ has a vertex cover of size at most $k$ if and only if $G$ admits an identification to  a star by at most $p+k-1$ identifications.
\end{observation}

\begin{proof}
The claim is trivial if $G$ has no edges. Assume from now on that this is not the case. 

Let $X$ be a vertex cover of size at most $k$. As $G$ has an edge, $X\neq\emptyset$.
We identify the vertices of $X$ into a single vertex, and then identify this vertex with the $p$ isolated vertices of the original graph. It is straightforward to see that the resulting graph is a star obtained by $|X|-1+p\leq p+k-1$ identifications.

For the opposite direction, assume that a star $S$ is an identification of $G$ obtained by at most $p+k-1$ identifications. Denote by $c$ the center vertex of $S$. Let $\Wcal=\{W(x)\mid x\in V(S)\}$ be an $S$-witness structure. By~\Cref{lem:singletons}, we can assume that for every $v\in V(S)$ distinct from $c$, $|W(v)|=1$. Then $|W(c)|\leq p+k$ and the isolated vertices of $G$ are in $W(c)$. Let $X\subseteq W(c)$ be the set of non-isolated vertices of $W(c)$. We have that $|X|\leq k$, and because for each leaf $v$ of $S$, its bag $W(v)$ is a singleton, $X$ is a vertex cover of $G$. This completes the proof.  
\end{proof}

The equivalence between \textsc{Vertex Cover} and  \kidH{Star},
established in~\Cref{obs:vc}, immediately implies that all \FPT and kernelization results for \textsc{Vertex Cover} (see~\cite{CyganFKLMPPS15,DowneyF13}) can be rewritten for \kidH{Star} parameterized by $k$. We leave this to the reader and consider more interesting cases of \Hid{$\Hcal$} when $\Hcal$ is either the class of forests or the class of trees. For this, we use the following lemma. To state it, denote by ${\sf Id}_F(G)$ the minimum number of identifications needed to obtain a forest from $G$. 

\begin{lemma}\label{lem:cut_comp}
Let $G$ be a graph, and let $(A,B)$ be a separation of $G$ of order at most one.
Suppose that $F$ is a forest such that $F$ can be obtained from $G$ using ${\sf Id}_F(G)$ identifications. 
Then for any $F$-witness structure $\Wcal=\{W(x)\mid  x\in V(F)\}$, 
there are exactly $|A\cap B|$ vertices $x\in V(F)$ with $W(x)\cap A\neq\emptyset$ and $W(x)\cap B\neq\emptyset$. Furthermore, $A\cap B\subseteq W(x)$ if $A\cap B\neq\emptyset$. 
\end{lemma}

\begin{proof}
Let  $\Wcal=\{W(x)\mid  x\in V(F)\}$ be an arbitrary $F$-witness structure. 
Clearly, if $A\cap B\neq\emptyset$, that is, $A\cap B=\{w\}$ for some $w\in V(G)$, then there is $x\in V(F)$ such that $A\cap B\subseteq W(x)$. Therefore, it is sufficient to show that 
there are at most $|A\cap B|$ vertices $x\in V(F)$ with $W(x)\cap A\neq\emptyset$ and $W(x)\cap B\neq\emptyset$.

For the sake of contradiction, assume that there is $y\in V(F)$ such that $W(y)$ contains vertices of both $A$ and $B$ and the unique vertex of $A\cap B$ is not in $W(y)$ if $A\cap B\neq\emptyset$. Let $X_1=\{x\in V(F)\mid W(x)\cap A\neq\emptyset\}$ and 
$X_2=\{x\in V(F)\mid W(x)\cap B\neq\emptyset\}$. 
Consider $\Wcal_1=\{W_1(x)=W(x)\cap A\mid x\in X_1\}$ and $\Wcal_2=\{W_2(x)=W(x)\cap B\mid x\in X_2\}$.
Note that, by the construction, the sets $W_i(y)$ for $i\in[2]$ are disjoint. Furthermore, 
there are $|A\cap B|$ pairs of sets of $\Wcal_1$ and $\Wcal_2$ containing common vertices, and if 
$A\cap B\neq\emptyset$ then the unique vertex of $A\cap B$ is this common vertex. 
For $i\in[2]$, denote by $F_i$ the graph with the vertex set $X_i$ such that distinct $x_1,x_2\in X_i$ are adjacent if and only if $W_i(x_1)$ and $W_i(x_2)$ are adjacent in $G$. By the definition, $F_1$ and $F_2$ are subgraphs of $F$ and, therefore, are forests.

Suppose that $A\cap B=\emptyset$. Then the disjoint union of copies of $F_1$ and $F_2$ is an identification of $G$, since $G$ has no edges with their endpoints in $A$ and $B$. 
Also, we have that $\Wcal_1\cup \Wcal_2$ is the set of bags in the corresponding witness structure.
This is a forest obtained by 
$\sum_{x\in X_1}(|W_1(x)|-1)+\sum_{x\in X_2}(|W_2(x)|-1)$ identifications. However,
$\sum_{x\in X_1}(|W_1(x)|-1)+\sum_{x\in X_2}(|W_2(x)|-1)<\sum_{x\in V(F)}(|W(x)|-1)$ because the bag $W(y)$ contains vertices of both $A$ and $B$. This contradicts the assumption that the forest $F$ is an identification of $G$ requiring the minimum number of identifications. 

Assume that $A\cap B=\{w\}$ for $w\in V(G)$. Let $x_1\in X_1$ and $x_2\in X_2$ be the vertices of $F_1$ and $F_2$, respectively, such that $w\in W_1(x_1)$ and $w\in W_2(x_2)$. Consider the forest $F'$ obtained from the disjoint copies of $F_1$ and $F_2$ by identifying the copies of $x_1$ and $x_2$. Because $G$ has no edges with their endpoints in $A\setminus\{w\}$ and $B\setminus\{w\}$, $F'$ is the identification of $G$ with witness structure 
$(\Wcal_1\setminus \{W_1(x_1)\})\cup(\Wcal_2\setminus\{W_2(x_2)\})\cup\{W_1(x_1)\cup W_2(x_2)\}$.
The number of identifications is
$\sum_{x\in X_1\setminus\{x_1\}}(|W_1(x)|-1)+\sum_{x\in X_2\setminus\{x_2\}}(|W_2(x)|-1)+
|W_1(x_1)\cup W_2(x_2)|-1=
\sum_{x\in X_1\setminus\{x_1\}}(|W_1(x)|-1)+\sum_{x\in X_2\setminus\{x_2\}}(|W_2(x)|-1)+|W(z)|-1
$, where $z\in V(F)$ is a vertex of $F$ such that $w\in W(z)$. Because $y\neq z$ and $W(y)$ contains vertices of both $A\setminus\{w\}$ and $B\setminus\{w\}$,    
$\sum_{x\in X_1\setminus\{x_1\}}(|W_1(x)|-1)+\sum_{x\in X_2\setminus\{x_2\}}(|W_2(x)|-1)+|W(z)|-1<\sum_{x\in V(F)}(|W(x)|-1)$ contradicting the choice of $F$. 

In both cases, we got a contradiction. This proves the lemma.
\end{proof}

In particular, \Cref{lem:cut_comp} immediately implies the following two properties. 

\begin{lemma}\label{lem:sep_one}
Let $G$ be a graph, and let $(A,B)$ be a separation of $G$ of order one. Let also $k\geq 0$
be an integer. Then the instances $(G,k)$ and $(G',k)$ of \kidH{Forest} are equivalent, where $G'$ is the disjoint union of the copies of $G[A]$ and $G[B]$. 
 \end{lemma}

\begin{lemma}\label{lem:union}
Let $G$ be a graph with connected components $G_1,\ldots,G_s$. Then 
${\sf Id}_F(G)={\sf Id}_F(G_1)+\dots+{\sf Id}_F(G_s)$.
\end{lemma}

For \kidH{Forest}, we show that the problem admits a polynomial kernel when parameterized by $k$.

\begin{theorem}\label{th:kernel_forest}
\kidH{Forest} admits a  kernel with $\Ocal(k^3)$ vertices 
when parameterized by $k$.
\end{theorem}

\begin{proof}
Let $(G,k)$ be an instance of \kidH{Forest}.
We exhaustively apply the following reduction rules. 

\begin{reduction}\label{R1}
Delete all bridges and isolated vertices in $G$.
\end{reduction}

\begin{claim}\label{cl:del_bridges}
The instances $(G,k)$ and $(G',k)$, where $G'$ is the graph obtained after applying \Cref{R1},  are equivalent instances of \kidH{Forest}.
\end{claim}

\begin{cproof}
Since $G'$ is a subgraph of $G$ and the class of forests is closed under taking subgraphs, if $(G,k)$ is a \yes-instance, then $(G',k)$ is a \yes-instance as well. For the opposite direction, assume that $(G',k)$ is a \yes-instance.

Suppose first that $G'$ is constructed from $G$ by deleting isolated vertices.
Then, because an edgeless graph is a forest, we have that
${\sf Id}_F(G)={\sf Id}_F(G')$ by~\Cref{lem:union}. Then $(G,k)$ is a \yes-instance of \kidH{Forest}.

Assume now that $G'$ is constructed by deleting bridges, and denote by $G_1,\ldots,G_s$ the connected components of $G'$. By~\Cref{lem:union}, 
${\sf Id}_F(G')={\sf Id}_F(G_1)+\dots+{\sf Id}_F(G_s)$. For each $i\in[s]$, let $F_i$ be a forest obtained from $G_i$ by ${\sf Id}_F(G_i)$ identifications, and let $F$ be the disjoint union of $F_1,\ldots,F_s$. We have that $F$ is an identification of $G'$ obtained by ${\sf Id}_F(G')$ identifications. Consider the corresponding $F$-witness structure $\Wcal=\{W(x)\mid x\in V(F)\}$.
For every bridge $uv$ of $G$, there are $x_u,x_v\in V(F)$ in distinct connected components of $F$ such that $u\in W(x_u)$ and $v\in W(x_v)$. Let $F'$ be the graph obtained from $F$ by making $x_u$ and $x_v$ adjacent for all bridges $uv$. It is straightforward to see that $F'$ is a forest. Furthermore, $\Wcal$ is an $F'$-witness structure of $G$. Thus, ${\sf Id}_F(G)\leq {\sf Id}_F(G')$ and, therefore, 
$(G,k)$ is a \yes-instance of \kidH{Forest}.

Since the deletion of isolated vertices and bridges can be done successively, the above arguments prove that $(G,k)$ is a yes-instance. This completes the proof.
\end{cproof}

If $G$ becomes empty after applying \Cref{R1}, then before the application of the rule $G$ was a forest. Thus, in this case,
we return a trivial \yes-instance of constant size and stop.

\begin{reduction}\label{R2}
For any two distinct vertices $u,v\in V(G)$, if $|N_G(u)\cap N_G(v)|>k+1$, then identify $u$ with $v$ and decrease $k$ by one.
\end{reduction}

\begin{claim}\label{cl:ident_common}
The instances $(G,k)$ and $(G',k-1)$, where $G'$ is the graph obtained after applying \Cref{R2},  are equivalent instances of \kidH{Forest}.
\end{claim}

\begin{cproof}
Obviously, if $(G',k-1)$ is a \yes-instance, then $(G,k)$ is a \yes-instance.
Reciprocally, let $\Wcal=\{W(x)\mid x\in V(F)\}$ be an $F$-witness structure of $G$ for some forest $F$ obtained from $G$ by at most $k$ identifications.
Let $x_u,x_v\in V(F)$ be the vertices such that $u\in W(x_u)$ and $v\in W(x_v)$. We claim that 
$x_u=x_v$.
Suppose for the sake of contradiction that $x_u\ne x_v$.
Let $a$ (resp. $b$) be the number of vertices of $N_G(u)\cap N_G(v)$ in $W(x_u)$ (resp. $W(x_v)$).
Let $R$ be the set of remaining vertices of $N_G(u)\cap N_G(v)$, i.e., those that are neither identified with $u$ nor $v$.
If $R=\emptyset$, then $a+b\geq k+2$ and 
$(|W(x_u)|-1)+(|W(x_v)|-1)\geq a+b\geq k+2$, contradicting that $F$ is obtained by at most $k$ identifications. Suppose that $R\neq\emptyset$.
Given that $|N_G(u)\cap N_G(v)|\ge k+2$, $|R|\ge k+2-a-b$.
If there are two distinct witness bags $W(r_1)$ and $W(r_2)$ containing vertices of $R$, then $r_1,x_u,r_2,x_v$ is a $4$-cycle in $F$, a contradiction.
Therefore, there is a vertex $r\in V(F)$ such that $R\subseteq W(r)$.
However, the total number of identifications in this case is at least 
$(|W(x_u)|-1)+(|W(x_v)|-1)+(|W(r)|-1)\geq a+b+|R|-1\geq k+1$.
Therefore, identifying to $F$ requires at least $k+1$ identifications; a contradiction.
Thus, if $(G,k)$ is a \yes-instance, then $u$ and $v$ are identified. Since this identification can be done first, we have that $(G',k-1)$ is a \yes-instance. This proves the claim.
\end{cproof}

We immediately stop and return a trivial no-instance of \kidH{Forest} of constant size if $k$ becomes negative after applying \Cref{R2}. From now on, we assume that this is not the case.

We claim that after exhaustively applying \Cref{R1} and \Cref{R2}, either the obtained graph has bounded size, or we have a \no-instance.

\begin{claim}\label{cl:upper}
    If $(G,k)$ is a \yes-instance to which \Cref{R1} and \Cref{R2} cannot be applied, then $|V(G)|\le 2k+(k+1)\cdot\binom{2k}{2}$.
\end{claim}

\begin{cproof}
Let $(G,k)$ be a \yes-instance of \kidH{Forest} to which \Cref{R1} and \Cref{R2} cannot be applied, and let $F$ be a forest obtained from $G$ by performing at most $k$ identifications. 
Denote by $S$ the set of vertices of $G$ that are identified with some other vertices, that is, belong to non-trivial bags of the $F$-witness structure. Since the total number of identifications is at most $k$, $|S|\leq 2k$.
For any edge $e$ in $G$, given that $e$ is part of a cycle (because $e$ is not a bridge), at least one of its endpoints has to be identified with some other vertex. This means that $S$ is a vertex cover of $G$ of size at most $2k$.  
Because $G$ has no isolated vertices and bridges by \Cref{R1}, every vertex of $G$ has degree at least~$2$. In particular, every $v\in I$ of the independent set $I=V(G)\setminus S$ has at least two neighbors in $S$. Observe that because \Cref{R2} is not applicable, for any distinct $u,v\in S$, there are at most $k+1$ common neighbors of these vertices. Therefore, $|I|\leq (k+1)\cdot\binom{|S|}{2}$, and $|V(G)|=|S|+|I|\leq 2k+(k+1)\cdot\binom{2k}{2}$. This proves the claim.
\end{cproof}

Using \Cref{cl:upper}, we return a trivial \no-instance if the graph $G$ obtained by 
the exhaustive application of \Cref{R1} and \Cref{R2} has more than $2k+(k+1)\cdot\binom{2k}{2}$ vertices. Otherwise, we return $(G,k)$.

The correctness of the kernelization algorithm follows from \Cref{cl:del_bridges}, \Cref{cl:ident_common}, and \Cref{cl:upper}. Since bridges can be found in linear time by the standard depth-first search algorithm, \Cref{R1} can be applied in linear time. Since
\Cref{R2} can be trivially applied in $\Ocal(n^3)$ time, the overall running time is polynomial. This proves that \kidH{Forest} admits a kernel with $\Ocal(k^3)$ vertices.
\end{proof}

Using \Cref{th:kernel_forest}, we prove a similar result for \kidH{Tree}.

\begin{corollary}\label{cor:trees}
\kidH{Tree} admits a kernel with $\Ocal(k^3)$ vertices when parameterized by $k$.
\end{corollary}

\begin{proof}
We reduce the problem to \kidH{Forest} 
using the following claim.

\begin{claim}\label{cl:trees}
An instance $(G,k)$ of \kidH{Tree} for a graph $G$ with $s$ connected components is equivalent to the instance $(G,k-s+1)$ of \kidH{Forest}.    
\end{claim}

\begin{cproof}
Let $G$ be a graph with connected components $G_1,\ldots,G_s$. The claim holds for $s=1$ by \Cref{obs:connect}. Assume that $s\geq 2$.

Suppose that $(G,k)$ is a \yes-instance of \kidH{Tree}. Then there is a sequence of at most $k$ identifications of the vertices of $G$ resulting in a tree $T$. 
Because $G$ has $s$ connected components, this sequence contains $s-1$ identifications resulting in a connected graph $G'$. Then $G'$ can be identified to a forest by at most $k-s+1$ identifications. Denote by $S$ the vertices of $G'$ obtained by identifying distinct vertices of $G$. 
Notice that for each vertex $v\in S$, there is a separation $(A,B)$ of $G'$ with $A\cap B=\{v\}$. Then by~\Cref{lem:sep_one}, if $(G',k-s+1)$ is a \yes-instance of \kidH{Forest}, then $(G,k-s+1)$ is a \yes-instance. 

For the opposite implication, assume that $(G,k-s+1)$ is a \yes-instance of \kidH{Forest}.
Denote by $G_1,\ldots,G_s$ the connected components of $G$. 
By~\Cref{lem:union}, ${\sf Id}_F(G)={\sf Id}_F(G_1)+\dots+{\sf Id}_F(G_s)$.
For $i\in[s]$, denote by $F_i$ the forest obtained from $G_i$ by ${\sf Id}_F(G_i)$ identifications. By~\Cref{obs:connect}, each $F_i$ is a tree. 
Let $F$ be the disjoint union of $F_1,\ldots,F_s$. We have that $F$ is an identification of $G$ obtained by ${\sf Id}_F(G)\leq k-s+1$ identifications. Finally, we observe that we can obtain a tree from $F$ by $s-1$ identifications. Thus, $(G,k)$ is a \yes-instance of \kidH{Tree}. This proves the claim.
\end{cproof}

Let $(G,k)$ be an instance of \kidH{Tree}. Denote by $G_1,\ldots,G_s$ the connected components of $G$. If $k+1<s$, we conclude that $(G,k)$ is a \no-instance by \Cref{cl:trees}. In this case, we return a trivial \no-instance of \kidH{Tree} and stop. Otherwise, we apply the kernelization algorithm from~\Cref{th:kernel_forest} for $(G,k-s+1)$. It returns an instance $(G',k')$ of \kidH{Forest}, where $G'$ has $\Ocal(k^3)$ vertices. By~\Cref{cl:trees}, $(G,k)$ and $(G',k')$ are equivalent. If $G'$ is a connected graph, then because of~\Cref{obs:connect}, $(G',k')$ is an equivalent instance of \kidH{Tree}. However, $G'$ may be disconnected. In this case, consider the connected components $G_1',\ldots,G_t'$ of $G'$. For each $i\in[t]$, let $v_i$ be an arbitrary vertex of $G_i'$. Denote by $G''$ the graph obtained by identifying $v_1,\ldots,v_t$ into a single vertex $v$. By~\Cref{lem:sep_one}, the instances $(G',k')$ and $(G'',k')$ of \kidH{Forest} are equivalent. Since $G''$ is connected, $(G'',k')$ is an equivalent instance of \kidH{Tree}. 
Therefore, our kernelization algorithm returns $(G'',k')$. This completes the description of the kernelization algorithm and its correctness proof. Since we can find the connected components in linear time by the standard breadth-first search, the running time is polynomial. This proves the corollary.
\end{proof}

\Cref{th:kernel_forest} and \Cref{cor:trees} imply that \Hid{$\Hcal$} can be solved in $k^{\Ocal(k)}\cdot n^{\Ocal(1)}$ time when $\Hcal$ is the class of forests or trees, as we can solve the problem by brute force on the reduced instances by guessing identifications.

\begin{corollary}\label{cor:FPT_F_T}
Let $\Hcal$ be the class of forests or trees.
Then \Hid{$\Hcal$} can be solved in $k^{\Ocal(k)}\cdot n^{\Ocal(1)}$ time. 
\end{corollary}

\subsection{Identification to a Given Forest}
By~\Cref{cor:given_forest}, \id is \W{1}-hard on instances $(G,F)$ where $F$ is a forest, when parameterized by the number of leaves of $F$. In this section, we complement this result by demonstrating an \XP algorithm for this parameterization. First, we show the result when the input graph $G$ is connected.
Note that any identification of a connected graph $G$ is connected by~\Cref{obs:connect}. Therefore, we can assume that the input forest $F$ is a tree.

\begin{lemma}\label{lem:tree}
\id is solvable in time $\Ocal(n^{\ell-1}\cdot (n+m))$ on instances $(G,T)$ where $G$ is a connected graph and $T$
is a tree with $\ell$ leaves. 
\end{lemma}

\begin{proof}
Let $G$ be a connected graph, and let $T$ be a tree. We assume that $|V(T)|\leq |V(G)|=n$, as otherwise, we have a trivial \no-instance. We also assume that $\ell\geq 2$ because for the single-vertex $T$, the problem is trivial. 
We select an arbitrary leaf $r$ of $T$ as a \emph{root} vertex which defines the parent-child relation on $V(T)$. For $x\in V(T)$, we denote by $T_x$ the subtree of $T$ induced by the descendants of $x$ (including $x$). Given a $T$-witness structure $\Wcal=\{W(x)\mid x\in V(T)\}$ of $G$, we set $W_x=\bigcup_{y\in V(T_x)}W(y)$ for $x\in V(T)$.

Let $L$ be the set of leaves of $T$ in the rooted tree; note that $|L|=\ell-1$ as the root is not a leaf of the rooted tree $T$. 
Our algorithm constructs a special $T$-witness structure $\Wcal=\{W(x)\mid x\in V(T)\}$ bottom-up starting from $x\in L$.
We say that a $T$-witness structure $\Wcal=\{W(x)\mid x\in V(T)\}$ is \emph{regular} if 
(i)~$|W(x)|=1$ for $x\in L$ and (ii)~for every non-root vertex $y\in V(T)\setminus L$,
$W(y)=N_G(\bigcup_{x\in C_y}W_x)$ where $C_y\subseteq V(T)$ is the set of children of $y$ in $T$.  
\begin{claim}\label{cl:regular}
If $G$ can be identified to $T$, then $G$ has a regular $T$-witness structure.     
\end{claim}

\begin{cproof}
Since the vertices of $L$ are pairwise non-adjacent, there is a $T$-witness structure where each leaf-bag indexed by a vertex of $L$ is a singleton by~\Cref{lem:singletons}. For each $T$-witness structure of this type, property (i) is fulfilled. Among all $T$-witness structures satisfying (i), we select a $T$-witness structure $\Wcal=\{W(x)\mid x\in V(T)\}$ such that the sum of distances in $T$ from $r$ to the non-root vertices $y\in V(T)\setminus L$ with $W(y)\neq N_G(\bigcup_{x\in C_y}W_x)$ is minimum. We claim that $\Wcal$ satisfies~(ii).

To obtain a contradiction, assume that $y\in V(T)\setminus L$ is a non-root vertex at maximum distance from~$r$ for which $W(y)\neq N_G(\bigcup_{x\in C_y}W_x)$. Since $y\neq r$, $y$ has the unique parent $z$. Since $T$ is a tree, we have that $N_G(\bigcup_{x\in C_y}W_x)\subseteq W(y)$ by the definition of a witness structure. 
Also, each neighbor $v$ of a vertex $u\in W(y)$ is in $N_G[W(y)]$, that is, either $v\in W(z)$ or $v\in W(t)$ for $t\in C_y\cup\{y\}$. Consider 
$\Wcal'=\{W'(x)\mid x\in V(T)\}$ where $W'(y)=N_G(\bigcup_{x\in C_y}W_x)$, $W'(z)=W(z)\cup(W(y)\setminus W'(y))$, and $W'(x)=W(x)$ for all $x\neq y,z$. 
Since $G$ is connected, for every $x\in C_y\cup\{z\}$, there are adjacent $v_x\in W'(x)$ and $v_y\in W'(y)$. Indeed, for $x\in C_y$, this holds because $W(x)\subseteq W_x$ has a neighbor in $W(y)$, and every such neighbor belongs to $N_G(\bigcup_{x'\in C_y}W_{x'})=W'(y)$. For $x=z$, observe that every edge of $G$ with exactly one endpoint in $W'(y)\cup\bigcup_{x'\in C_y}W_{x'}$ has this endpoint in $W'(y)$ and the other endpoint in $W(z)\cup (W(y)\setminus W'(y))=W'(z)$; such an edge exists because $G$ is connected and $W'(z)\neq\emptyset$. Then, because $N_G[W(y)\setminus W'(y)]\subseteq W(y)\cup W(z)$, we have that $\Wcal'$ is a $T$-witness structure. However, $W'(y)=N_G(\bigcup_{x\in C_y}W_x)$ contradicting the choice of $\Wcal$. Thus, $\Wcal$ satisfies~(i) and (ii). This completes the proof.
\end{cproof}
    
Based on~\Cref{cl:regular}, our algorithm finds a regular $T$-witness structure (if it exists). For each $x\in L$, we guess a vertex $v_x$ of $G$ such that $W(x)=\{v_x\}$. As $|L|=\ell-1$, there are at most $n^{\ell-1}$ possibilities to guess the vertices $v_x$. For each choice, we verify whether it can be extended to a regular $T$-witness structure. 
If we find such a choice, we return the corresponding solution. Otherwise, we report a \no-instance. 
                             
From now on, we assume that the bags $\{v_x\}$ for $x\in L$ are given. To extend the witness structure, we apply the standard bottom-up dynamic programming approach. For every non-root vertex $y\in V(T)\setminus L$ such  that for all children $x$ of $y$ in $T$ (and, therefore, for their descendants), the bags are constructed, we define
$W(y)=N_G(\bigcup_{x\in C_y}W_x)$ where $C_y\subseteq V(T)$ is the set of children of $y$ in $T$.
Finally, we set $W(r)=V(G)\setminus\bigcup_{x\in V(T)\setminus\{r\}}W(x)$. Since $T$ has at most $n$ vertices, the family $\Wcal=\{W(x)\mid x\in V(T)\}$ can be constructed in $\Ocal(n+m)$ time for a given set of leaf-bags. In the last step, we verify in $\Ocal(n+m)$ time whether $\Wcal$ is a $T$-witness structure using the definition. 

This completes the description of the algorithm. Its correctness follows from \Cref{cl:regular}. Since initially we consider at most $n^{\ell-1}$ choices of $\ell-1$ vertices of $G$ in the leaf-bags, the overall running time is $\Ocal(n^{\ell-1}(n+m))$. This concludes the proof.
\end{proof}

In~\Cref{thm:lin_forest_hard}, we proved that \id is \W{1}-hard on instances $(G,L)$ where $L$ is a linear forest, when parameterized by the number of connected components of $L$. The proof was based on the close relation of our problem to \textsc{Bin Packing}. Using this relation further, we show that the problem is in \XP when parameterized by the number of connected components of $L$. We use this result as a building block for the algorithm for \id on instances $(G,F)$ where $F$ is a forest.

\begin{lemma}\label{lem:linear_forest_XP}
\id can be solved in $\Ocal(n^{t+2})$
time on instances $(G,L)$ where $L$ is a linear forest with $t$ connected components.    
\end{lemma}

\begin{proof}
Let $G$ be a graph with $s$ connected components $G_1,\ldots,G_s$, and let $P_1,\ldots,P_t$ be the connected components of $L$. We assume that $|V(L)|\leq |V(G)|=n$, $|E(L)|\leq |E(G)|=m$, and $t\leq s$ as, otherwise, we have a trivial \no-instance of the problem. Also, if $L$ is an edgeless graph and $t\leq s$, then $G$ can be identified to $L$ 
in a straightforward way.
From now on, assume that these cases do not apply.
For $i\in[s]$, let $d_i=\diam(G_i)$. Let also $\ell_i$ be the length of $P_i$ for $i\in[t]$.
We have the following property.

\begin{claim}\label{cl:diam_bin}
The graph $G$ can be identified to $L$ if and only if there is a partition $I_1,\ldots,I_t$ of $[s]$ such that for every $i\in[t]$, $\sum_{j\in I_i}d_j\geq\ell_i$.    
\end{claim}

\begin{cproof}
Suppose that $G$ can be identified to $L$. Then by~\Cref{obs:connect}, there is a partition $I_1,\ldots,I_t$ of $[s]$ such that for every $i\in[t]$, $P_i$ is an identification of the union of $G_j$ for $j\in I_i$. Then by~\Cref{prop:path}, for every $i\in [t]$, $\ell_i\leq \sum_{j\in I_i}\diam(G_j)=\sum_{j\in I_i}d_j$. For the opposite implication, assume that 
there is a partition $I_1,\ldots,I_t$ of $[s]$ such that for every $i\in[t]$, $\sum_{j\in I_i}d_j\geq\ell_i$. Then by~\Cref{prop:path}, for each $i\in[t]$, the subgraph of $G$ consisting of the connected components $G_j$ with $j\in I_i$ can be identified to $P_i$. This immediately implies that $G$ can be identified to $L$. This completes the proof.
\end{cproof}

Using~\Cref{cl:diam_bin}, we construct a dynamic programming algorithm inspired by the standard dynamic programming for \textsc{Bin Packing} for a bounded number of bins. 
However, it is convenient to do the following preprocessing. Notice that if $\ell_i=0$ for some $i\in[t]$, that is, $P_i$ is a trivial path, then we can select $I_i=\{j\}$ for $j\in[s]$, where $d_j=\min\{d_1,\ldots,d_s\}$. Thus, we may greedily select $I_i$ for trivial $P_i$. 
By this procedure, we reduce it to the case when $\ell_i\geq 1$ for $i\in[t]$. 
From now on, we assume that $\ell_i\geq 1$ for $i\in[t]$ and $s<n$.

For $i\in[0,s]$ and a $t$-tuple $(p_1,\ldots,p_t)$, where $p_j\in[0,\ell_j]$ for $j\in[t]$, we define the Boolean function $T(i,(p_1,\ldots,p_t))$ by setting 
$T(i,(p_1,\ldots,p_t))={\sf true}$ if there is a partition $I_1,\ldots,I_t$ of $[i]$ (where $[0]=\emptyset$), where some sets may be empty, such that for every $h\in[t]$,
$\sum_{j\in I_h}d_j\geq p_h$, and we set $T(i,(p_1,\ldots,p_t))={\sf false}$, otherwise. Here, we assume that for $I_i=\emptyset$, $\sum_{j\in I_i}d_j=0$.
Observe that $T(s,(\ell_1,\ldots,\ell_t))={\sf true}$ if and only if 
 there is a partition $I_1,\ldots,I_t$ of $[s]$ such that for every $i\in[t]$, $\sum_{j\in I_i}d_j\geq\ell_i$; in particular, each $I_i\neq\emptyset$ because $\ell_i\geq 1$.
Our dynamic programming algorithm computes the table of values of $T(i,(p_1,\ldots,p_t))$ for all $i\in[0,s]$ and all $t$-tuples $(p_1,\ldots,p_t)$.

We successively compute the tables for $i=0,\ldots,s$. Initially, for $i=0$, we set
$T(0,(p_1,\ldots,p_t))={\sf true}$ for $p_1=\dots=p_t=0$, and $T(0,(p_1,\ldots,p_t))={\sf false}$, otherwise. For $i\geq 1$, we use the following recurrence:
\begin{equation}\label{eq:rec}
T(i,(p_1,\ldots,p_t))=\bigvee_{j\in[t]}T(i-1,(p_1,\ldots,p_{j-1},
\max\{0,p_j-d_i\},p_{j+1},\ldots,p_t)).    
\end{equation}

The correctness of computing $T(0,(p_1,\ldots,p_t))$ immediately follows from the definition of the function. The correctness of the recurrence~(\ref{eq:rec}) is proved by the standard arguments using the observation that for each $i\in[s]$, $i$ should be in some $I_j$ for $j\in[t]$. Since $\ell_j\leq n-1$ for $j\in[t]$,  
the table of values of $T(i,(p_1,\ldots,p_t))$ for each $i\in[0,s]$ can be computed in $\Ocal(n\cdot n^{t})$ time by (\ref{eq:rec}). Because $s<n$, the values
of $T(s,(p_1,\ldots,p_t))$ can be computed in $\Ocal(n^{t+2})$ time.
As the connected components of $G$ can be found in linear time and the diameters of the components can be computed in $\Ocal(nm)$ time, the overall running time is $\Ocal(n^{t+2})$. 
This concludes the proof.
\end{proof}

Now we show that \id is in \XP on instances $(G,F)$ where $F$ is a forest, when parameterized by the number of leaves of $F$.

\begin{theorem}\label{thm:forest_XP}
\id is solvable in $n^{\Ocal(\ell^2)}$ time on instances $(G,F)$ where $F$ is a forest with $\ell$ leaves.
\end{theorem}

\begin{proof}
Let $G$ be a graph with $s$ connected components $G_1,\ldots,G_s$, and let $T_1,\ldots,T_t$ be the connected components of a forest $F$ with $\ell$ leaves. As in the proof of~\Cref{lem:linear_forest_XP}, we assume that $|V(F)|\leq |V(G)|=n$, $|E(F)|\leq |E(G)|=m$, and $t\leq s$ to rule out trivial \no-instances. Furthermore, by~\Cref{lem:linear_forest_XP}, \id can be solved in $\Ocal(n^{\ell+2})$ time when $F$ is a linear forest since the number of leaves of a linear forest is at least the number of its connected components. Thus, we may assume that this is not the case. In particular, this means that $\ell\geq 3$.

Our algorithm is based on structural properties of solutions. Suppose that $G$ can be identified to $F$. For $I\subseteq[s]$, we denote by $G_I$ the subgraph of $G$ composed of the connected components $G_i$ for $i\in I$.

Observe that $G$ can be identified to $F$ if and only if there is $I\subseteq [s]$ such that  $G_I$ can be identified to $F$. The forward implication is trivial, and if $F$ is an identification of $G_I$ for some $I\subseteq[s]$, an identification of $G$ can be constructed by identifying the vertices of $G_i$ for $i\in[s]\setminus I$ with an arbitrary vertex of $F$. We assume that $I\subseteq [s]$ is a set of minimum size such that $G_I$ can be identified to $F$. Denote by $\Wcal=\{W(x)\mid x\in V(F)\}$ the corresponding $F$-witness structure.

Given that $F$ has $\ell\geq 3$ leaves, $F$ has at most $2\ell-2$ vertices of degree $\ne2$. 
Let $V_{\ne 2}$ be these vertices. 
For a vertex $x\in V_{\ne 2}$ and a leaf $y$ of the connected component $T_j$ of $F$ containing $x$, we say that a connected component $G_i$ for $i\in I$ is \emph{$y$-important for $x$} if
(i)~$W(x)\cap V(G_i)\neq \emptyset$ and (ii)~the distance between $y$ and the closest to $y$ vertex $z\in V(T_j)$ with  $W(z)\cap V(G_i)\neq\emptyset$ is minimum, where the minimum is taken over all $G_i$ containing a vertex of $W(x)$. For $x\in V_{\neq 2}$, we 
denote by ${\sf Imp}(x)\subseteq I$ an inclusion-minimal set such that for each leaf $y$ of $T_j$, there is a $y$-important for $x$ connected component $G_i$ 
with $i\in{\sf Imp}(x)$. Note that because every leaf of $T_j$ is a leaf of $F$,
$|{\sf Imp}(x)|\leq \ell$ for $x\in V_{\neq 2}$.
We show the following property.

\begin{claim}\label{cl:imp}
Let $G_i$ for $i\in I$ be a connected component such that $V(G_i)$ has a vertex in $W(x)$ for some $x\in V_{\neq 2}$. Then $i\in{\sf Imp}(x')$ for some $x'\in V_{\neq 2}$.
\end{claim}
 
\begin{cproof}
The proof is by contradiction. Assume that there is $i\in I$ and $x\in V_{\neq 2}$ such that 
$W(x)\cap V(G_i)\neq \emptyset$ and 
$i\notin{\sf Imp}(x')$ for any $x'\in V_{\neq 2}$. 
Let $T_j$ for $j\in[t]$ be the connected component of $F$ containing $x$. Note that by~\Cref{obs:connect}, only bags $W(z)$ for $z\in V(T_j)$ may contain vertices of $G_i$.
If $T_j$ has a single vertex, then ${\sf Imp}(x)$ contains a unique element and, by the minimality of $I$, the claim holds. Assume that $T_j$ has edges. 
Consider an arbitrary 
edge $z_1z_2\in E(T_j)$ such that $W(z_1)\cap V(G_i)\neq\emptyset$ and 
$W(z_2)\cap V(G_i)\neq\emptyset$. Because $T_j$ is a tree, $z_1z_2$ is an edge of an $x$-$y$ path for some leaf $y$ of $T_j$.  
Because $i\notin{\sf Imp}(x)$, there is a $y$-important for $x$ connected component $G_h$ with $h\in{\sf Imp}(x)$. Then by~\Cref{obs:connect}, there are $v_1\in W(z_1)$ and $v_2\in W(z_2)$ such that $v_1,v_2\in V(G_h)$ and $v_1v_2\in E(G_h)$. Moreover, removing the vertices of $G_i$ does not make any bag empty. Indeed, $W(x)$ contains a vertex of $G_h$ for $h\in{\sf Imp}(x)$, where ${\sf Imp}(x)\neq\emptyset$ and $i\notin{\sf Imp}(x)$. Furthermore, because $G_i$ is connected, the vertices $z\in V(T_j)$ with $W(z)\cap V(G_i)\neq\emptyset$ induce a subtree of $T_j$ containing $x$. Hence, for every such $z\neq x$, the edge of the $x$-$z$-path in $T_j$ incident to $z$ is an edge $z_1z_2$ as above, and we obtain that $W(z)$ contains a vertex of a connected component $G_h$ with $h\neq i$. However, this implies that
$\Wcal'=\{W(z)\setminus V(G_i)\mid z\in V(F)\}$ is an $F$-witness structure for $G_{I'}$ where $I'=I\setminus\{i\}$ contradicting the minimality of $I$. This proves the claim.
\end{cproof}

Let $I'=\bigcup_{x\in V_{\neq 2}}{\sf Imp}(x)$ and let $I''=I\setminus I'$. 
For $x\in V(F)$, let $W'(x)=W(x)\cap(\bigcup_{i\in I'}V(G_i))$.
Consider the subgraph $F'$ of $F$ with $V(F')=\{x\in V(F)\mid W'(x)\neq\emptyset\}$ such that distinct $x,y\in V(F')$ are adjacent if and only if the sets $W'(x)$ and $W'(y)$ are adjacent.
By definition, we have that $F'$ is an identification of $G_{I'}$. Furthermore,
the construction and \Cref{obs:connect} immediately imply the following properties:
\begin{itemize}
\item[(i)] for each $x\in V_{\neq 2}$, there is a connected component of $F'$ containing $x$ and its neighbors in $F$, and 
\item[(ii)] each connected component of $F'$ contains a vertex $x\in V_{\neq 2}$ and its neighbors in $F$.
\end{itemize}

Besides (i) and (ii), we have the next useful property of the connected components of $F'$.
Let $x,y\in V_{\neq 2}$ be distinct vertices in the same connected component of $F'$ such that
the unique $x$-$y$-path $P$ in $F'$ does not contain other vertices of $V_{\neq 2}$. We say that a vertex $z\in V(P)$ is an \emph{$x$-$y$-connector} if there are $i\in{\sf Imp}(x)$ and  $j\in{\sf Imp}(y)$ such that both $G_i$ and $G_j$ have a vertex in $W(z)$.

\begin{claim}\label{cl:gluing}
Let $x,y\in V_{\neq 2}$ be distinct vertices in the same connected component of $F'$ such that
the unique $x$-$y$-path $P$ in $F'$ does not contain other vertices of $V_{\neq 2}$. 
Then $P$ contains an $x$-$y$-connector $z$.
\end{claim}

\begin{cproof}
The claim is straightforward if there is $i\in{\sf Imp}(x)$ such that $W(y)$ contains a vertex of $G_i$ or, symmetrically, there is $i'\in{\sf Imp}(y)$ such that $W(x)$ contains a vertex of $G_{i'}$. Assume that this is not the case. Let $T_q$ for $q\in[t]$ be the connected component of $F$ containing $x$ and $y$. Let $a$ be a leaf of $T_q$ such that $y$ is a vertex of the $x$-$a$-path in $T_q$. Moreover, let $b$ be a leaf of $T_q$ such that $x$ is a vertex of the $y$-$b$-path in $T_q$. There is $i\in {\sf Imp}(x)$ such that $G_i$ is 
$a$-important for $x$, and there is $i'\in {\sf Imp}(y)$ such that $G_{i'}$ is 
$b$-important for $y$. Let $a'$ be the vertex of $T_q$ closest to $a$ such that $W(a')$ contains a vertex of $G_i$, and let $b'$ be the vertex of $T_q$ closest to $b$ such that $W(b')$ contains a vertex of $G_{i'}$.
Observe that because $W(y)\cap V(G_i)=\emptyset$ and $W(x)\cap V(G_{i'})=\emptyset$,     
$a'$ and $b'$ are internal vertices of $P$. We claim that the $x$-$a'$ and $y$-$b'$-paths $Q_1$ and $Q_2$ in $T_q$ have a common vertex~$z$.

For the sake of contradiction, assume that these paths have no common vertices. Then there is an edge $r_1r_2$ of $P$ such that $r_1r_2$ is not an edge of $Q_1$ and $Q_2$. Because $r_1r_2$ is an edge of $F'$, there is $r\in V_{\neq 2}\cap V(T_q)$ and $h\in {\sf Imp}(r)$ such that $W(r_1)$ and $W(r_2)$ contain adjacent vertices of $G_h$. By symmetry, we can assume that $r$ is closer to $x$ than $y$. Then by~\Cref{obs:connect}, $G_h$ contains a vertex of $W(x)$. Because both $W(r_1)$ and $W(r_2)$ contain vertices of $G_h$, we obtain a contradiction with the fact that $G_i$ is $a$-important for $x$. This proves that  
$Q_1$ and $Q_2$ have a common vertex $z$.

It remains to observe that $W(z)$ contains vertices of both $G_i$ and $G_{i'}$, that is,
$z$ is an $x$-$y$-connector. This completes the proof.
\end{cproof}

Now, we deal with the edges of $F$ that are not in $F'$.
Let $R=E(F)\setminus E(F')$, and set $L=F[R]$. Then, because of properties (i) and (ii) of $F'$, we have that 
\begin{itemize}
\item[(a)] $L$ is a linear forest,
\item[(b)] each connected component of $L$ is a subpath of an $x$-$y$-path in $F$ with $x,y\in V_{\neq 2}$ whose endpoints are the vertices of $V(F')$.
\end{itemize}
In particular, by (b), the vertices of $L$ have degree~$2$ in $F$, and every $x$-$y$-path in $F$ with $x,y\in V_{\neq 2}$ whose internal vertices have degree $2$ in $F$ contains vertices of at most one connected component of $L$. 

We show that $L$ is an identification of $G_{I''}$. If $L$ is empty, then $F'=F$, because $F'$ contains all edges of $F$ and, by (i), all isolated vertices of $F$. Since $G_{I'}$ can be identified to $F'$, the minimality of $I$ implies that $I''=\emptyset$, and the claim holds. Assume that $L$ is non-empty, and let $z$ be an arbitrary vertex of $L$.
We construct the bags $W''(x)$ for $x\in V(L)$ as follows. Initially, set $W''(x)=\emptyset$ for every $x\in V(L)$. Consider $i\in I''$. By~\Cref{cl:imp}, $G_i$ has no vertices in the bags $W(x)$ for $x\in V_{\neq 2}$. Because $G_i$ is connected, there is a connected component $P_i$ of $F-V_{\neq 2}$ containing every vertex $y$ of $F$ with $W(y)\cap V(G_i)\neq\emptyset$. Notice that $P_i$ is a path and, by (i) and (ii), contains at most one connected component of $L$. If $P_i$ contains a connected component $Q_i$ of $L$, then for each $y\in V(P_i)$, we add the vertices of $W(y)\cap V(G_i)$ to $W''(x)$, where $x$ is the unique vertex of $Q_i$ closest to $y$. Otherwise, we add all vertices of $G_i$ to $W''(z)$.

By construction, the sets $W''(x)$ for $x\in V(L)$ are disjoint and their union is $V(G_{I''})$. Furthermore, the endpoints of each edge of $G_i$ for $i\in I''$ are placed either in the same bag or in two bags corresponding to adjacent vertices of $L$. Indeed, if $P_i$ contains $Q_i$, then their original bags correspond to equal or adjacent vertices of $P_i$, and the closest vertices of $Q_i$ to two adjacent vertices of $P_i$ are equal or adjacent. Otherwise, both endpoints are placed in $W''(z)$. Since distinct connected components of $G$ have no edges between them, distinct bags $W''(x)$ and $W''(y)$ can be adjacent only if $xy\in E(L)$.
Conversely, for each edge $y_1y_2\in E(L)$, there is $i\in I''$ and adjacent $v_1,v_2\in V(G_i)$ such that $v_1\in W(y_1)$ and $v_2\in W(y_2)$, because $y_1y_2\notin E(F')$. Then $y_1$ and $y_2$ belong to the connected component $Q_i$ of $L$ contained in $P_i$, and the construction places $v_1$ in $W''(y_1)$ and $v_2$ in $W''(y_2)$. Since $L$ has no isolated vertices, this also implies that $W''(x)\neq\emptyset$ for every $x\in V(L)$. Thus, $\Wcal''=\{W''(x)\mid x\in V(L)\}$ is an $L$-witness structure for $G_{I''}$.

We use the above properties to guess $I'$, $F'$, and $L$. To verify the correctness of each guess, we use the following claim.

\begin{claim}\label{cl:paths}
Let $F'$ be a subgraph of $F$ such that (i)~for each $x\in V_{\neq 2}$, there is a connected component of $F'$ containing $x$ and its neighbors in $F$, and (ii)~each connected component of $F'$ contains a vertex $x\in V_{\neq 2}$ and its neighbors in $F$. Let also 
$L=F[E(F)\setminus E(F')]$ be the linear forest. 
Suppose that there is a partition $I',I''$ of $[s]$, where $I''$ may be empty, such that $G_{I'}$ can be identified to $F'$, and $G_{I''}$
 can be identified to $L$. 
Then $G$ can be identified to $F$.
\end{claim}

\begin{cproof}
Assume that $\Wcal'=\{W'(x)\mid x\in V(F')\}$ is an $F'$-witness structure for $G_{I'}$, and let $\Wcal''=\{W''(x)\mid x\in V(L)\}$ be an $L$-witness structure for $G_{I''}$.
We set $W'(x)=\emptyset$ for $x\in V(F)\setminus V(F')$, and define $W''(x)=\emptyset$ for 
$x\in V(F)\setminus V(L)$. Then for $x\in V(F)$, define $W(x)=W'(x)\cup W''(x)$. 
Because the sets of indices $I'$ and $I''$ are disjoint, 
we find
that $\Wcal=\{W(x)\mid x\in V(F)\}$ is an $F$-witness structure for $G$.
This proves the claim. 
\end{cproof}

Now we are ready to construct our \XP-algorithm. 

First, we find the set of vertices $V_{\neq 2}\subseteq V(F)$. Recall that $|V_{\neq 2}|\leq 2\ell-2$. Then we guess $F'$ and $L$. We include all isolated vertices of $F$ in $F'$. 
Then we consider the paths $P$ in $F$ whose endpoints are the vertices of $V_{\neq 2}$ and whose internal vertices are of degree $2$. Notice that $F$ has at most $|V_{\neq 2}|-1\leq 2\ell-3$ such paths. For each $P$, we guess either zero or two distinct vertices $x$ and $y$ of degree~$2$. In the first case, the vertices and edges of $P$ are included in $F'$. In the second case, the $x$-$y$-subpath $P'$ of $P$ becomes a connected component of $L$. The edges of $E(P)\setminus E(P')$ become the edges of $F'$, and their endpoints are included in $V(F')$. Since the number of paths is upper-bounded by $2\ell-3$ and $|V(F)|\leq n$,
we have $n^{\Ocal(\ell)}$ possibilities to guess $F'$ and $L$. Note that $L$ may be empty.
We denote by $T_1',\ldots,T_r'$ the connected components of $F'$.

In the next step, we guess ${\sf Imp}(x)$ for $x\in V_{\neq 2}$. More precisely, for each $x\in V_{\neq 2}$, we guess a nonempty set of at most $\ell$ connected components $G_i$, whose indices form ${\sf Imp}(x)$, and,
for every $i\in{\sf Imp}(x)$, a vertex of $G_i$ intended to belong to the bag $W(x)$.
This choice defines ${\sf Imp}(x)$ and the set of vertices $U(x)\subseteq W(x)$ representing the connected components $G_i$ with $i\in{\sf Imp}(x)$.
Since $|V_{\neq 2}|\leq2\ell-2$ and $|{\sf Imp}(x)|\leq \ell$, we have $n^{\Ocal(\ell^2)}$ choices. This also defines $I'=\bigcup_{x\in V_{\neq 2}}{\sf Imp}(x)$; by~\Cref{cl:imp}, the connected components $G_i$ with $i\in[s]\setminus I'$ have no vertices in the bags $W(x)$ for $x\in V_{\neq 2}$. 

For each such choice, we check whether the sets $U(x)$ for $x\in V_{\neq 2}$ are disjoint and discard the choice if this is not the case. Otherwise, for each $x\in V_{\neq 2}$, we identify the vertices of $U(x)$. 
Then we consider the connected components $T_q'$ for $q\in[r]$ containing at least two vertices of $V_{\neq 2}$. For each pair $x$ and $y$ of distinct vertices of $V_{\neq 2}\cap V(T_q')$ such that the $x$-$y$-path in $T_q'$ does not contain other vertices of $V_{\neq 2}$, we use \Cref{cl:gluing}, and guess two vertices $v_x\in V(G_i)$ with $i\in{\sf Imp}(x)$ and $v_y\in V(G_j)$ with $j\in{\sf Imp}(y)$ such that $v_x,v_y\in W(z)$ for an $x$-$y$-connector $z$. Then we identify $v_x$ and $v_y$.  Since we have at most $2\ell-3$ such pairs $x,y$ in $F'$, the total number of guesses for connectors is $n^{\Ocal(\ell)}$.

Let $G'$ be the graph obtained from $G_{I'}$ by these identifications. Notice that for correct guesses, the vertices of $G'$ included in the bags of the vertices of the same connected component of $F'$ should be in the same connected component of $G'$. By~\Cref{obs:connect}, we have that the number of connected components of $G'$ should be the same as the number of connected components of $F'$. We check this and discard the current choice of identified vertices if this is not the case. Otherwise, denote by $G_1',\ldots,G_r'$ the connected components of $G'$. Observe that by the choice of ${\sf Imp}(x)$ for $x\in V_{\neq 2}$, we know which connected component of $G'$ has vertices in the bags of the vertices of each connected component of $F'$. Thus, we assume that for each $i\in[r]$, $T_i'$ 
is an identification of $G_i'$. We verify this in time $n^{\Ocal(\ell)}$ by calling the algorithm of~\Cref{lem:tree}  for $G_i'$ and $T_i'$ for every $i\in[r]$. 

Suppose that the algorithm of~\Cref{lem:tree} returns a positive answer for $G_i'$ and $T_i'$ for every $i\in[r]$. Then $G'$ can be identified to $F'$. By composing these identifications with those used to construct $G'$ from $G_{I'}$, we obtain an $F'$-identification of $G_{I'}$. If $L$ is empty, then $F'=F$, and we conclude that $G$ can be identified to $F$ as $I'\subseteq [s]$. Then we return the answer and stop.
If $L$ is non-empty,  we use the algorithm from~\Cref{lem:linear_forest_XP} for $G_{I''}$ and $L$, where $I''=[s]\setminus I'$. Since the number of connected components of $L$ is at most $2\ell-3$, this can be done in $n^{\Ocal(\ell)}$ time. 
If the algorithm returns that 
$G_{I''}$ can be identified to $L$, then we conclude by~\Cref{cl:paths} that $G$ can be identified to $F$ and report the answer. Otherwise, we discard the current guess of  ${\sf Imp}(x)$ for $x\in V_{\neq 2}$. If we fail to find an identification for every guess, we conclude that we have a \no-instance. This completes the description of the algorithm and its correctness proof.

To evaluate the running time, observe that the total number of guesses we make is $n^{\Ocal(\ell^2)}$. Since the connected components of a graph can be found in linear time by standard tools and the algorithms from~\Cref{lem:tree} and \Cref{lem:linear_forest_XP} run in $n^{\Ocal(\ell)}$, we have that the overall running time is $n^{\Ocal(\ell^2)}$. This concludes the proof.
\end{proof}

\section{Identification to Cliques and Cluster Graphs}\label{sec:cliques}

In this section, we consider \Hid{$\Hcal$} where $\Hcal$ is a class of complete or cluster graphs.
We refer to these problems as \kidH{Clique} and \kidH{Cluster}, respectively.

\subsection{Algorithmic Lower Bounds}
We begin by establishing the following algorithmic and kernelization lower bound. 

\begin{theorem}\label{th:NPclique}
\kidH{Clique} is \NP-complete and does not admit a polynomial kernel when parameterized by $k$ unless \coNP$\subseteq$~\NPp. Furthermore, these results hold for connected input graphs. 
\end{theorem}

\begin{proof}
We present a reduction from the {\sc Set Cover} problem. The task of this problem is, given  a universe $\Ucal$, a family $\Scal$ of non-empty subsets of $\Ucal$ such that $\Ucal=\bigcup_{S\in \Scal}S$, and $k\in\bN$, to decide whether  there is a subfamily $\Scal'\subseteq \Scal$ of size at most $k$ such that $\bigcup_{S\in\Scal'}S=\Ucal$. In fact, we may assume without loss of generality that $k\leq |\Scal|$. Hence, by adding arbitrary sets from $\Scal\setminus\Scal'$ if needed, we may simply assume that $|\Scal'|=k$. The problem {\sc Set Cover} is well-known to be \NP-complete~\cite{GareyJ79}. Furthermore, {\sc Set Cover} parameterized by the number $m$ of sets has no polynomial kernel~\cite{DomLS09inco,CyganFKLMPPS15} unless
\coNP$\subseteq$~\NPp. Since our reduction is a polynomial parameter transformation, we simultaneously prove both claims of the theorem.

\begin{figure}[ht]
\center
\includegraphics[scale=1]{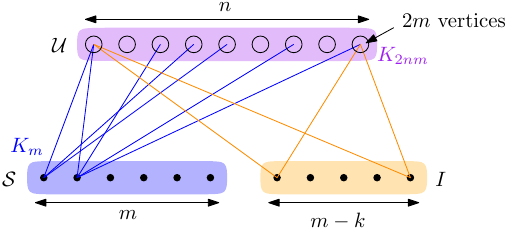}
\caption{The graph $G$ constructed in the proof of \Cref{th:NPclique}.}
\label{fig:NPclique}
\end{figure}

Let $\Ucal$ be a universe with $n$ elements, $\Scal$ be a family of $m$ subsets of $\Ucal$, and $k\in\bN$.
We construct a graph~$G$ on $2nm+2m-k$ vertices as follows (cf. \Cref{fig:NPclique} for an illustration).

\begin{itemize}
\item For each element $u\in\Ucal$, create a set $V_u$ containing $2m$ vertices.
\item For each subset $S\in\Scal$, create a vertex $v_S$.
\item Make a clique out of $V_\Ucal=\bigcup_{u\in\Ucal}V_u$ and out of $V_\Scal=\{v_S\mid S\in\Scal\}$.
\item For each $u\in \Ucal$ and each $S\in\Scal$, add an edge between $v_S$ and each vertex of $V_u$ if and only if $u\in S$.
\item Construct an independent set $I$ of size $m-k$ and add an edge between each vertex of $I$ and each vertex of $V_u$, for all $u\in\Ucal$.
\end{itemize}
Notice that $G$ is connected. 
We claim that $(\Ucal,\Scal,k)$ is a \yes-instance of \textsc{Set Cover} if and only if $(G,m-1)$ is a \yes-instance of \kidH{Clique}.

Suppose that $(\Ucal,\Scal,k)$ is a \yes-instance of {\sc Set Cover}, and let $\Scal'\subseteq \Scal$ be a family of size~$k$ such that $\bigcup_{S\in\Scal'}S=\Ucal$. We identify the vertices in $G$ corresponding to $\Scal'$ into a single vertex (this requires $k-1$ identifications). Since $\Scal'$ is a set cover, the obtained vertex is adjacent to every vertex of $V_\Ucal$. Then we identify pairwise each vertex corresponding to a set in $\Scal\setminus \Scal'$ with a vertex of $I$ (this requires $m-k$ identifications).
The obtained graph is a clique.
So, $(G,m-1)$ is a \yes-instance of \kidH{Clique}.

Suppose now that $(G,m-1)$ is a \yes-instance of \kidH{Clique}.
Let $K_p$ for $p\geq 2nm+m-k+1$ be the clique obtained after at most $m-1$ identifications, and let $\Wcal=\{W(x)\mid x\in V(K_p)\}$ be a $K_p$-witness structure of $G$.
In the next two claims, we show that there is a bag $W(x)\subseteq V_\Scal$ in $\Wcal$ corresponding to a set cover of size at most $k$.

\begin{claim}
There exists $x\in V(K_p)$ such that $W(x)\subseteq V_\Scal$.
\end{claim}

\begin{cproof}
For the sake of contradiction, assume that for each $v\in V_\Scal$, there is $u\in V(G)\setminus V_\Scal$ such that $u$ and $v$ are in the same bag of $\Wcal$. This means that each vertex $v\in V_\Scal$ is identified with a vertex outside $V_\Scal$, and because $|V_\Scal|\geq m$, we need at least $m$ identifications, contradicting that the number of identifications is at most $m-1$. This proves the claim.
\end{cproof}

\begin{claim}\label{cl:small}
We have $|W(x)|\le k$.
\end{claim}

\begin{cproof}
Notice that for each $v\in I$, there is a vertex $u\in V(G)\setminus I$ such that $u$ and $v$ are in the same bag of $\Wcal$. Otherwise, if there is $y\in V(K_p)$ such that $W(y)\subseteq I$, the bags $W(x)$ and $W(y)$ are not adjacent, contradicting that we obtain a clique by the identifications. Thus, we need at least $|I|=m-k$ identifications for the vertices of $I$. This implies that to identify the vertices in $W(x)$ into a single vertex, we can use at most $(m-1)-(m-k)=k-1$ identifications. Therefore, $|W(x)|\leq k$.  
\end{cproof}

Let $\Scal'=\{S\in\Scal\mid v_S\in W(x)\}$.
Given that we do at most $m-1$ identifications, for each $u\in \Ucal$, there is at least one vertex $v_u$ out of the $2m$ vertices of $V_u$ such that $W(y_u)=\{v_u\}$ for some $y_u\in V(K_p)\setminus\{x\}$.
Given that $x$ is adjacent to each $y_u$, there is $v_S\in W(x)$ such that $v_S$ is adjacent to $v_u$ in $G$. This implies that the element $u\in\Ucal$ is covered by $S$.
Thus, $\Scal'$ covers $\Ucal$ and $(\Ucal,\Scal,k)$ is a \yes-instance of {\sc Set Cover}. This completes the proof.
\end{proof}

Because the lower bounds from \Cref{th:NPclique} hold for connected graphs, by making use of~\Cref{obs:connect}, we obtain the same results for cluster graphs. 

\begin{corollary}\label{cor:NPcluster}
\kidH{Cluster} is \NP-complete and does not admit a polynomial kernel when parameterized by $k$ unless \coNP$\subseteq$~\NPp.
\end{corollary}

\subsection{FPT Algorithms and Kernelization for Identification to Cliques and Cluster Graphs}

In this section, we provide algorithmic results for \kidH{Clique} and \kidH{Cluster}
parameterized by $k$ and the dual parameterization by $n-k$.
In~\Cref{th:NPclique} and \Cref{cor:NPcluster}, we proved that it is unlikely that \kidH{Clique} and \kidH{Cluster} have polynomial kernels when parameterized by $k$. However, the problems are \FPT for this parameterization. We first demonstrate such an algorithm for 
\kidH{Clique} and then use it for \kidH{Cluster}.

We consider the auxiliary {\sc $S$-Constrained Clique-Identification} problem. This problem asks, given a graph $G$, a set of vertices $S\subseteq V(G)$ such that 
$G-S$ is a clique, and an integer $k$, whether $G$ can be identified to a clique $K_h$ for $h\geq n-k$ with a $K_h$-witness structure $\Wcal=\{W(x)\mid x\in V(K_h)\}$, where each bag contains at most one vertex of $S$. In other words, we ask whether we can obtain a clique by at most $k$ identifications in such a way that the vertices of $S$ are not identified with each other. This variant of \kidH{Clique} has a polynomial kernel when parameterized by the size of $S$ and $k$.

\begin{lemma}\label{lem:aux_kernel}
{\sc $S$-Constrained Clique-Identification} has a kernel with at most $(|S|+1)(k+2)$ vertices
and can be solved in $((|S|+1)(k+2))^{\Ocal(k)}\cdot n^{\Ocal(1)}$ time.
\end{lemma}

\begin{proof}
Let $(G,S,k)$ be an instance of \textsc{$S$-Constrained Clique-Identification}.
If we have $|V(G)|\leq |S|(k+2)+k$, we return the input instance, as it has a bounded size. Assume that this is not the case and $|V(G)|\geq|S|(k+2)+k+1$.
The following claim will be useful for us.

\begin{claim}\label{cl:samebag}
Suppose that $v\in S$ has at least $k+1$ non-adjacent vertices in $V(G)\setminus S$. 
Then in any $K_h$-witness structure $\Wcal=\{W(x)\mid x\in V(K_h)\}$ for $h\geq n-k$ where each bag contains at most one vertex of $S$, there is $x\in V(K_h)$ such that $W(x)$ contains $v$ and a vertex of $V(G)\setminus S$.
\end{claim}

\begin{cproof}
The proof is by contradiction. 
Suppose that there is a $K_h$-witness structure $\Wcal=\{W(x)\mid x\in V(K_h)\}$ such that the bag $W(x)$ containing $v$ does not include any vertex of $V(G)\setminus S$. Let $X\subseteq V(G)\setminus S$ be the set of non-neighbors of $v$ that are not in $S$. Since $vu\notin E(G)$ for $u\in X$, we have that each bag $W(y)$ containing a vertex of $X$ should contain at least one vertex in $V(G)\setminus X$. However, this means that we need at least $|X|\geq k+1$ identifications, contradicting that $h\geq n-k$. This proves the claim.  
\end{cproof}

For each $v\in S$, we construct the set $R_v$ of at most $k+1$ non-neighbors of $v$ in $V(G)\setminus S$---if $v$ has at most $k$ such non-neighbors, we include all of them in $R_v$, and we include arbitrary $k+1$ non-neighbors, otherwise. We say that the vertices of $R_v$ are \emph{marked}. After marking vertices for $v\in S$, we select an arbitrary set 
$U\subseteq V(G)\setminus(S\cup(\bigcup_{v\in S}R_v))$ of size $k+1$ and mark these vertices as well; note that $U$ exists because $|V(G)|\geq|S|(k+2)+k+1$. Now we delete all non-marked vertices of $V(G)\setminus S$. Denote by $G'$ the obtained graph. Note that 
$|V(G')|\leq |S|+\sum_{v\in S}|R_v|+|U|\leq |S|(k+2)+k+1$. We show that the instances $(G,S,k)$ and $(G',S,k)$ of \textsc{$S$-Constrained Clique-Identification} are equivalent.

\begin{claim}\label{cl:equiv_aux}
$(G,S,k)$ and $(G',S,k)$ are equivalent instances of \textsc{$S$-Constrained Clique-Identification}.
\end{claim}

\begin{cproof}
Let $v$ be a non-marked vertex of $G-S$. Assume that $G'=G-v$. As non-marked vertices can be deleted one by one, it is sufficient to show that $(G,S,k)$ and $(G',S,k)$ are equivalent. 

Suppose that $(G,S,k)$ is a yes-instance of \textsc{$S$-Constrained Clique-Identification}, and let $\Wcal=\{W(x)\mid x\in V(K_h)\}$ be a $K_h$-witness structure for $h\geq n-k$. Let $y\in V(K_h)$ be the vertex such that $v\in W(y)$. Since $|U|\geq k+1$, there is $z\in V(K_h)$ distinct from $y$ 
such that $W(z)$ contains a vertex $u\in U$ but $W(z)\cap S=\emptyset$.
Indeed, otherwise every vertex of $U$ is either in $W(y)$, which also contains $v\notin U$, or in a bag containing a vertex of $S$. Then every bag containing $t\geq 1$ vertices of $U$ has at least $t+1$ vertices, and we need at least $|U|\geq k+1$ identifications, contradicting that $h\geq n-k$.
Consider the family of sets $\Wcal'=\{W'(x)\mid x\in V(K_h)\setminus \{y\}\}$, where 
$W'(z)=(W(y)\cup W(z))\setminus\{v\}$ and $W'(x)=W(x)$ for $x\in V(K_h)$ distinct from $y$ and $z$. 
Observe that because $W(z)\cap S=\emptyset$, each bag $W'(x)$ contains at most one vertex of $S$.
We show that $\Wcal'$ is a $(K_h-y)$-witness structure for $G'$. 
By the construction $\Wcal'$ is a partition of the vertex set of $G'$. Thus, we have to show that the sets $W'(x)$ are pairwise adjacent. 
As $W'(x)=W(x)$ for $x\neq y,z$, we have that $W'(x)$ and $W'(x')$ are adjacent if both $x$ and $x'$ are distinct from $y$ and $z$ because $\Wcal$ is a $K_h$-witness structure.  
Since $W(z)\subseteq W'(z)$, $W'(z)$ and $W'(x)=W(x)$ for $x\neq y,z$ are adjacent.   
This 
proves that $\Wcal'$ is a $(K_h-y)$-witness structure for $G'$. Since $|V(G')|=n-1$ and $K_h-y$ has at least $h-1$ vertices, we obtain that 
$(G',S,k)$ is a \yes-instance of \textsc{$S$-Constrained Clique-Identification}.

For the opposite implication, assume that $(G',S,k)$ is a yes-instance of \textsc{$S$-Constrained Clique-Identification}, and let $\Wcal'=\{W'(x)\mid x\in V(K_{h-1})\}$ be a  corresponding $K_{h-1}$-witness structure for $h-1\geq (n-1)-k$. Assume that $K_{h-1}=K_h-y$ for some $y\in V(K_h)$. We define 
$\Wcal=\{W(x)\mid x\in V(K_{h})\}$ by setting $W(x)=W'(x)$ for $x\in V(K_{h-1})$, and defining $W(y)=\{v\}$. We claim that $\Wcal$ is a $K_h$-witness structure for $G$.
Note that $\Wcal$ is a partition of $V(G)$. Because $\Wcal'$ is a $K_{h-1}$-witness structure for $G'$, the sets $W(x)=W'(x)$ for $x\in V(K_{h-1})$ are pairwise adjacent. If $W(x)=W'(x)$ for $x\in V(K_{h-1})$ contains a vertex of $V(G')\setminus S$, then $W(y)$ and $W(x)$ are adjacent because $V(G)\setminus S$ is a clique. Assume that $W(x)$ does not contain a vertex of $V(G')\setminus S$. Then $W(x)=\{w\}$ for $w\in S$, and, by~\Cref{cl:samebag}, $w$ has at most $k$ non-neighbors in $V(G')\setminus S$.  Thus, $w$ has at most $k$ non-neighbors in $V(G)\setminus S$, as otherwise $k+1$ of them would be in $R_w\subseteq V(G')$, and all these non-neighbors are marked. Since $v$ is non-marked, we have that $v$ is a neighbor of $w$ in $G$. Therefore, $W(x)$ and $W(y)$ are adjacent. This proves that $\Wcal$ is a $K_h$-witness structure for $G$. Hence, $(G,S,k)$ is a \yes-instance of \textsc{$S$-Constrained Clique-Identification}. This proves the claim.
\end{cproof}

Because $(G,S,k)$ and $(G',S,k)$ are equivalent and $|V(G')|\leq |S|(k+2)+k+1$, our kernelization algorithm returns $(G',S,k)$.
Since marking of the vertices of $G$ can be done in polynomial time, $G'$ can be constructed in polynomial time. This proves that \textsc{$S$-Constrained Clique-Identification} has a polynomial kernel with at most $(|S|+1)(k+2)$ vertices.

To see that the problem can be solved in $((|S|+1)(k+2))^{\Ocal(k)}\cdot n^{\Ocal(1)}$ time, note that given an instance where $n\le (|S|+1)(k+2)$, we can use brute force using the fact that there are at most $\binom{n}{2}^k$ possibilities to perform $k$ identifications.
This completes the proof.
\end{proof}

Now we prove the main result of the subsection.

\begin{theorem}\label{th:FPTclique}
\kidH{Clique} is solvable in time $k^{\Ocal(k)}\cdot n^{\Ocal(1)}$.
\end{theorem}

\begin{proof}
Let $(G,k)$ be an instance of \kidH{Clique}. 

Observe that if $(G,k)$ is a \yes-instance, then we can delete at most $2k$ vertices of $G$ to obtain a clique---given at most $k$ identifications, we simply delete every vertex which is identified with some other vertex. 
Since a set of vertices $X$ is a clique in $G$ if and only if $X$ is an independent set in the complement graph $\overline{G}$, we have that in the case of a \yes-instance, $\overline{G}$ has a vertex cover of size at most $2k$. Because \textsc{Vertex Cover} is \FPT parameterized by the solution size, we can in $2^{\Ocal(k)}\cdot n^{\Ocal(1)}$ time (see, e.g.,~\cite{ChenKX06impr,CyganFKLMPPS15}), either 
find $S\subseteq V(G)$ of size at most $2k$ such that $G-S$ is a clique, or report that $(G,k)$ is a \no-instance of \kidH{Clique}.

Assume that there is $S\subseteq V(G)$ of size at most $2k$ such that $G-S$ is a clique.
If $(G,k)$ is a \yes-instance, then there is a $K_h$-witness structure 
$\Wcal=\{W(x)\mid x\in V(K_h)\}$ for $h\geq n-k$. 
In the next step of our algorithm, we guess the restriction of $\Wcal$ to $S$.
In other words, we guess the partition $\Scal=\{S_i\mid i\in[p]\}$ of $S$ such that the vertices of $S_i$ are in the same bag, and the vertices of distinct $S_i$ and $S_j$ are in distinct bags. Note that since our budget is upper bounded by $k$, we consider only partitions 
with $\sum_{i=1}^p(|S_i|-1)\leq k$. Given $\Scal$, we identify the vertices in the same set of the partition. 

Let $G'$ be the obtained graph, and let $S'$ be the set of vertices obtained from $S$ by these identifications. We set $k'=k-\sum_{i=1}^p(|S_i|-1)$. Observe that if $(G,k)$ is a \yes-instance and the guess of $\Scal$ is correct, then  
we can obtain a clique from $G'$ by at most $k'$ identifications in such a way that 
distinct vertices of $S'$ are in the distinct bags of the corresponding witness structure. This means that if $(G,k)$ is a \yes-instance and $\Scal$ is the restriction of a $K_h$-witness structure to $S$, 
then $(G',S',k')$ is a \yes-instance of \textsc{$S$-Constrained Clique-Identification}. 
We use the algorithm from~\Cref{lem:aux_kernel} to verify whether this is the case. 
Observe that because $G'$ is obtained by $k-k'$ identifications, if $(G',S',k')$ is a \yes-instance of \textsc{$S$-Constrained Clique-Identification}, then $(G,k)$ is a \yes-instance of \kidH{Clique}.  This completes the description of the algorithm and its correctness proof.

To evaluate the running time, recall that $S$ is constructed in $2^{\Ocal(k)}\cdot n^{\Ocal(1)}$ time. Then we consider $k^{\Ocal(k)}$ partitions of the set $S$ of size at most $2k$. Finally, the algorithm from~\Cref{lem:aux_kernel} runs in $k^{\Ocal(k)}\cdot n^{\Ocal(1)}$ time as $|S|\leq 2k$. Thus, the overall running time is $k^{\Ocal(k)}\cdot n^{\Ocal(1)}$. This completes the proof.
\end{proof}

\begin{corollary}\label{cor:FPTcluster}
\kidH{Cluster} can be solved in $k^{\Ocal(k)}\cdot n^{\Ocal(1)}$ time.
\end{corollary}

\begin{proof}
We reduce  \kidH{Cluster} to \kidH{Clique} by guessing the connected components of the input graph identified to separate cliques of a cluster graph.
Let $(G,k)$ be an instance of \kidH{Cluster}. Denote by $\Ccal$ the set of connected components of $G$.

If $\Ccal$ contains at least $k+1$ connected components that are not cliques, then we conclude that $(G,k)$ is a \no-instance because at least two vertices of such a component should be involved in the identifications in any solution. Assume that this is not the case. 
Denote by $\Ccal'$ the connected components of $G$ distinct from cliques, and let $G'$ be the subgraph of $G$ with these connected components. 
We show that the identifications can be restricted to $G'$.

\begin{claim}\label{cl:restr_marked}
The instances $(G,k)$ and $(G',k)$ are equivalent.    
\end{claim}

\begin{cproof}
Suppose that $(G,k)$ is a \yes-instance of \kidH{Cluster}. We choose $H$ to be a cluster graph such that $G$ can be identified to $H$ by at most $k$ identifications, and select an $H$-witness structure $\Wcal=\{W(x)\mid x\in V(H)\}$ such that the number of vertices in the non-trivial bags, that is, bags $W(x)$ of size at least two, is minimum.
We show that for every bag $W(x)$ with at least two vertices, $W(x)\subseteq V(G')$.

\begin{figure}[ht]
    \centering
    \includegraphics[scale=0.8]{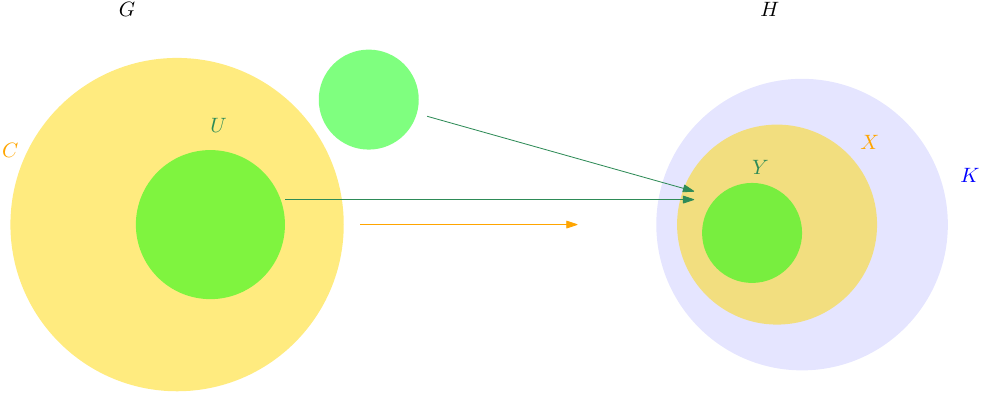}
    \caption{Illustration for \Cref{cl:restr_marked}.}
    \label{fig:idCluster}
\end{figure}

Suppose that there is a  connected component $C\in \Ccal\setminus \Ccal'$ containing vertices of non-trivial bags. Then $C$ is a complete graph. Let $X=\{x\in V(H)\mid W(x)\cap V(C)\neq\emptyset\}$. Because $C$ is a clique, $X$ is a clique of $H$. Since $H$ is a cluster graph, there is a clique $K$, that is, a connected component of $H$, such that $X\subseteq V(K)$. Then 
$\Wcal'=\{W(x)\mid x\in V(K)\}$ is a $K$-witness structure. 
Denote by $U\subseteq V(C)$ the set of 
vertices of $C$ that are in the non-trivial bags containing the vertices of other connected components of $G$, and let $Y=\{x\in V(H)\mid W(x)\cap U\neq\emptyset\}$. 
See \Cref{fig:idCluster} for an illustration.
We have that $Y\neq\emptyset$. 
Otherwise, by the choice of $H$ and $\Wcal$, we would have that $K=C$, and $W(x)$ for $x\in V(K)$ would be trivial bags because $C$ is a clique. We modify $\Wcal'$ as follows. 
We replace the bags $W(x)$ with $x\in Y$ by the single bag $W'(Y)=\bigcup_{x\in Y}W(x)\setminus U$, and the bags $W(x)$ with $W(x)\subseteq V(C)$ are replaced by the trivial bags $W'(v)$ for $v\in V(C)$. 
Notice that these sets compose a witness structure for the disjoint union of $C$ and the clique $K'$ obtained from $K$ by identifying the vertices of $Y$ and deleting the vertices $x$ with $W(x)\subseteq V(C)$, which is a cluster graph. 
Since for every $x\in Y$, $W(x)$ contains at least one vertex of $U$,
$|W'(Y)|-1=|\bigcup_{x\in Y}W(x)\setminus U|-1\leq \sum_{x\in Y}(|W(x)|-1)$, and 
the number of identifications needed to obtain this graph does not exceed the number of identifications for $K$.
However, the number of vertices in the nontrivial bags is smaller. This contradicts the choice of $H$ and $\Wcal$.

Thus, for every non-trivial bag $W(x)$ with $x\in V(H)$, $W(x)\subseteq V(G')$.
Consider the subgraph $H'$ induced by the vertices $x\in V(H)$ with $W(x)\cap V(G')\neq\emptyset$, and let $\Wcal'=\{W(x)\cap V(G')\mid x\in V(H')\}$. Notice that $H'$ is a cluster graph. Since 
$W(x)\subseteq V(G')$ for every non-trivial bag $W(x)$, we have that $\Wcal'$ is an $H'$-witness structure. Therefore, $(G',k)$ is a \yes-instance of \kidH{Cluster}.

The opposite implication is straightforward, as $G$ is the disjoint union of $G'$ and the cluster graph consisting of the connected components of $G$ in $\Ccal\setminus \Ccal'$. This completes the proof. 
\end{cproof}

Because $G'$ has at most $k$ connected components, there are at most $k^{\Ocal(k)}$ possibilities to partition $\Ccal'$ into subsets $\Ccal_1,\ldots,\Ccal_s$ such that 
(i)~for each $i\in[s]$, the subgraph $G_i$ consisting of the connected components of $G$ in $\Ccal_i$ is identified to a clique $C_i$, and (ii)~the cluster graph $H$ with the connected components $C_1,\ldots,C_s$ is an identification of $G'$ obtained by using at most $k$ identifications. 
We consider all possible partitions by guessing the pairs of connected components of $\Ccal'$ 
whose vertices are going to be identified in a potential solution. Then for every $i\in[s]$, we use the algorithm from \Cref{th:FPTclique} to find the minimum $0\leq k_i\leq k$ such that $(G_i,k_i)$ is a \yes-instance of \kidH{Clique} if such $k_i$ exists. If we have 
such $k_1,\ldots,k_s\leq k$ and 
$k_1+\dots+k_s\leq k$, we conclude that $(G',k)$ and, therefore, $(G,k)$ is a \yes-instance of \kidH{Cluster}. Otherwise, we discard the current choice of the partition of $\Ccal'$. If we fail to find a partition certifying 
that $(G,k)$ is a \yes-instance, we conclude that $(G,k)$ is a \no-instance.
This concludes the description of the algorithm and its correctness proof.

Note that finding the connected components of $G$ and their marking can be done in linear time.
Furthermore, we consider $k^{\Ocal(k)}$ partitions of $\Ccal'$. For each partition, 
we solve the corresponding instances of \kidH{Clique} in $k^{\Ocal(k)}\cdot n^{\Ocal(1)}$ time using the algorithm from \Cref{th:FPTclique}. Then the overall running time is 
$k^{\Ocal(k)}\cdot n^{\Ocal(1)}$. This concludes the proof.
\end{proof}

Now we consider the dual parameterization by $n-k$. In this case, \kidH{Clique} and 
\kidH{Cluster} admit polynomial kernels. We also show that the same holds for \id on instances $(G,H)$ where $H$ is a cluster graph for the parameterization by $p=|V(H)|$. Our kernelization algorithms for these problems are based on almost the same arguments. Therefore, we collect them in one theorem.

\begin{theorem}\label{th:kernelclique}
\kidH{Clique} and \kidH{Cluster}
admit kernels with $\Ocal(p^4)$ vertices when parameterized by $p=n-k$. Furthermore, \id on instances $(G,H)$, where $H$ is a cluster graph with $p$ vertices, admits a kernel with $\Ocal(p^6)$ vertices.
\end{theorem}

\begin{proof}
We first present the algorithm for \kidH{Clique} and then explain how to modify it for \kidH{Cluster} and \id with $H$ being a cluster graph.

Let $(G,k)$ be an instance of \kidH{Clique}. 
Given that the class of cliques is closed under identifications, $(G,k)$ is a \yes-instance if and only if $G$ can be identified to $K_p$ for $p=n-k$. If $p\leq 1$, then we return a trivial \yes-instance of constant size. If $p\geq 2$ and $G$ has no edges, then $(G,k)$ is a \no-instance. 
From now on, we assume that $p\geq 2$ and $G$ has at least one edge. 

We first greedily compute an arbitrary inclusion-maximal matching $M$ of $G$. We show that if $M$ is sufficiently big, then $(G,k)$ is a \yes-instance.

\begin{claim}\label{cl:matching}
    If $|M|\ge\binom{p}{2}$, then $(G,k)$ is a \yes-instance of \kidH{Clique}.
\end{claim}

\begin{cproof}
We set $\ell=\binom{p}{2}$.
Let $\{x_1y_1,\dots,x_\ell y_\ell\}$ be the edges of $K_p$ and $\{u_1v_1,\dots,u_\ell v_\ell\}$ be the $\ell$ first edges of $M$.
For each $x\in V(K_p)$, we initially set $W(x)=\{u_i\mid i\in[\ell]\text{ s.t. }x=x_i\}\cup\{v_i\mid i\in[\ell]\text{ s.t. }x=y_i\}$. 
We additionally add all remaining vertices of $G$ distinct from $u_i$ and $v_i$ for $i\in[\ell]$ to $W(x^*)$ for some arbitrary $x^*\in V(K_p)$.
Then $\{W(x)\mid x\in V(K_p)\}$ is a $K_p$-witness structure in $G$, hence the result.  
\end{cproof}

Hence, if $|M|\geq\binom{p}{2}$, we return a trivial \yes-instance of \kidH{Clique} of constant size and stop. 

From now on, we assume that this is not the case and $|M|<\binom{p}{2}$. Since $M$ is inclusion-maximal, the set $S$ of the endpoints of the edges of $M$ is a vertex cover of $G$ of size at most $2\cdot\binom{p}{2}-2$. Let $I=V(G)\setminus S$.   
For each $v\in S$, we construct the set $R_v$ of at most $\binom{p}{2}$ neighbors of $v$ in $I$---if $v$ has at most $\binom{p}{2}$ such neighbors, we include all of them in $R_v$, and we include arbitrary $\binom{p}{2}$ neighbors, otherwise. We say that the vertices of $R_v$ are \emph{marked}. We set $G'=G[S\cup\big(\bigcup_{v\in S}R_v\big)]$ and define 
$k'=k-(|V(G)|-|V(G')|)$. Notice that $|V(G')|\leq |S|+|S|\binom{p}{2}\leq p^4$ and $p=n-k=|V(G')|-k'$.
We claim that the instances $(G,k)$ and $(G',k')$ are equivalent. 

\begin{claim}\label{cl:mark_clique}
$(G,k)$ and $(G',k')$ are equivalent instances of \kidH{Clique}.
\end{claim}

\begin{cproof}
Notice that we have to show that $G$ can be identified to $K_p$ if and only if $G'$ can be identified to $K_p$.

Suppose that $G$ can be identified to $K_p$. Let $\Wcal=\{W(x)\mid x\in V(K_p)\}$ be a $K_p$-witness structure. Because $K_p$ is a clique, for any two distinct $x,y\in V(K_p)$, there are adjacent $v_x\in W(x)$ and $v_y\in W(y)$. Thus, there is a set $A$ of $\binom{p}{2}$ edges of $G$ such that for distinct $x,y\in V(K_p)$, there are $v_x\in W(x)$ and $v_y\in W(y)$
with $v_xv_y\in A$. We say that the edges of $A$ are \emph{important}. The choice of a $K_p$-witness structure and important edges is not unique. We select 
$\Wcal$ and $A$ in such a way that the size of $A\cap E(G')$ is maximum. 
We prove that $A\subseteq E(G')$.

The proof is by contradiction. Assume that there are distinct $s,t\in V(K_p)$ such that for $v_sv_t\in A$ with $v_s\in W(s)$ and $v_t\in W(t)$, $v_sv_t\notin E(G')$. Then one of $v_s$ and $v_t$ is in $S$ and the other in $I$. By symmetry, we assume that $v_s\in S$ and $v_t\in I$. Note that $v_t\notin R_{v_s}$. Since $v_t$ is a non-marked neighbor of $v_s$, $|R_{v_s}|=\binom{p}{2}$. Each important edge has at most one endpoint in $I$. Then because $|R_{v_s}|=\binom{p}{2}$ and $v_t\notin R_{v_s}$, there is a vertex $v_z\in R_{v_s}$ in some bag $W(z)$ such that $v_z$ is not incident with any important edge. If $z=t$, we set $\Wcal'=\Wcal$. If $z\neq t$, 
we define $\Wcal'=\{W'(x)\mid x\in V(K_p)\}$ by setting $W'(t)=W(t)\cup\{v_z\}$, $W'(z)=W(z)\setminus\{v_z\}$, and $W'(x)=W(x)$ for $x\neq t,z$, that is, we move $v_z$ from $W(z)$ to $W(t)$.
Since the endpoints of the important edges are not moved, we have that $\Wcal'$ is a $K_p$-witness structure. However, now we can choose $A'=(A\setminus\{v_sv_t\})\cup\{v_sv_z\}$ to be the set of important edges, contradicting the choice of $A$. The obtained contradiction proves that 
$A\subseteq E(G')$.

Because $A\subseteq E(G')$, we have that
$\Wcal'=\{W(x)\cap V(G')\mid x\in V(K_p)\}$ is a $K_p$-witness structure for $G'$. Therefore,
$G'$ can be identified to $K_p$.

To show the opposite implication, assume that $G'$ can be identified to $K_p$ and consider a $K_p$-witness structure $\Wcal'=\{W'(x)\mid x\in V(K_p)\}$ for $G'$.  Let $z\in V(K_p)$ be an arbitrary vertex. Then $\Wcal=\{W(x)\mid x\in V(K_p)\}$, where $W(z)=W'(z)\cup (V(G)\setminus V(G'))$ and $W(x)=W'(x)$ if $x\neq z$, is a $K_p$-witness structure. This shows that 
$G$ can be identified to $K_p$ and completes the proof.
\end{cproof}

As $(G,k)$ and $(G',k')$ are equivalent instances and $|V(G')|\leq p^4$, our kernelization algorithm returns the instance $(G',k')$. Because an inclusion-maximal matching $M$ can be greedily selected in linear time, and the sets $R_v$ of marked vertices can also be found in linear time, the overall running time is linear. This completes the proof for \kidH{Clique}.

\medskip
Now we consider \kidH{Cluster}. Let $(G,k)$ be an instance of the problem. Notice that $(G,k)$ is a \yes-instance if and only if $G$ can be identified to a cluster graph $H$ with $p$ vertices.
In the first step, we exhaustively apply the following reduction rule.

\begin{reduction}\label{red:isol}
If $G$ has an isolated vertex $v$ then set $G:=G-v$.
\end{reduction}

The rule is trivially safe, because the class of cluster graphs is closed under deletions and additions of isolated vertices. If by the application of the rule we get the empty graph or obtain that $p=|V(G)|-k\leq 1$, we immediately conclude that $(G,k)$ is a \yes-instance of \kidH{Cluster}. 
In this case,
our algorithm returns a trivial \yes-instance of constant size and stops. From now on, we assume that this is not the case.

We find an inclusion-maximal matching $M$ of $G$. By~\Cref{cl:matching}, if $|M|\geq \binom{p}{2}$, then $(G,k)$ is a \yes-instance of \kidH{Clique}. Therefore, $(G,k)$ is a \yes-instance of \kidH{Cluster}, and our algorithm returns a trivial \yes-instance of \kidH{Cluster} of constant size. Suppose that $|M|<\binom{p}{2}$. Then the set $S$ of the endpoints of the edges of $M$ is a vertex cover of $G$ of size at most $2\binom{p}{2}-2$.
Let $I=V(G)\setminus S$.

Now we mark some vertices of $I$. Exactly as for \kidH{Clique}, for each $v\in S$, we select the set $R_v$ of at most $\binom{p}{2}$ neighbors of $v$ in $I$ and mark the vertices of $R_v$.
For \kidH{Cluster}, we have to mark additional vertices. For this, we find a spanning tree $T$ of each connected component of $G$ and mark the non-leaf vertices of this tree in $I$. Notice that we have at most 
$|S|-1\leq 2\binom{p}{2}-3$ such vertices. For the set $R$ of marked vertices, we define
$G'=G[S\cup R]$ and set 
$k'=k-(|V(G)|-|V(G')|)$. Observe that because of marking of certain vertices of the spanning trees, two vertices $u,v\in V(G')$ are in the same connected component of $G'$ if and only if they are in the same connected component of $G$. Note also that
$|V(G')|\leq |S|+|S|\binom{p}{2}+|S|-1\leq |S|(\binom{p}{2}+2)\leq p^4$ and $p=n-k=|V(G')|-k'$.
We claim that the instances $(G,k)$ and $(G',k')$ are equivalent. 

\begin{claim}\label{cl:mark_cluster}
The instances $(G,k)$ and $(G',k')$ are equivalent instances of \kidH{Cluster}.
\end{claim}

\begin{cproof}
Suppose that $(G,k)$ is a \yes-instance of \kidH{Cluster}.
Then by~\Cref{obs:connect}, there is a partition of the set $\Ccal$ of connected components of $G$ into $\Ccal_1,\ldots,\Ccal_s$ for some $s\geq 1$ and integers $p_1,\ldots,p_s$ with $p_1+\dots+p_s=p$ such that for each $i\in[s]$, the subgraph $G_i$ of $G$ consisting of the connected components of $G$ in $\Ccal_i$ can be identified to $K_{p_i}$.
Observe that because $G$ has no isolated vertices, each connected component of $G$ contains a vertex of $S$. Since $S\subseteq V(G')$, for every $i\in[s]$, $G_i$ contains vertices of $G'$.
For $i\in[s]$, let $G_i'=G_i[V(G')\cap V(G_i)]$. Using the arguments from~\Cref{cl:mark_clique}, we have that $(G_i',|V(G_i')|-p_i)$ is a \yes-instance of \kidH{Clique} for every $i\in[s]$. Since $G'$ is the disjoint union of $G_1',\ldots,G_s'$ and $p_1+\dots+p_s=p$, we have that $(G',k')$ is a \yes-instance of \kidH{Cluster}.

Assume now that $(G',k')$ is a \yes-instance of \kidH{Cluster}.
Then there is a partition of the set $\Ccal'$ of connected components of $G'$ into $\Ccal_1',\ldots,\Ccal_s'$ for some $s\geq 1$ and integers $p_1,\ldots,p_s$ with $p_1+\dots+p_s=p$ such that for each $i\in[s]$, the subgraph $G_i'$ of $G'$ consisting of the connected components of $G'$ in $\Ccal_i'$ can be identified to $K_{p_i}$. 
For $i\in[s]$, consider the subgraph $G_i$ of $G$ composed of the connected components of $G$ having vertices in $G_i'$. 
Because each connected component of $G$ contains a vertex of $S$, the union of $G_1,\ldots,G_s$ is $G$. Furthermore, since two vertices $u,v\in V(G')$ are in the same connected component of $G'$ if and only if they are in the same connected component of $G$,
we have that the subgraphs $G_1,\ldots,G_s$ are disjoint. Thus, $G$ is the disjoint union of $G_1,\ldots,G_s$. Since each $G_i'$ can be identified to $K_{p_i}$, the same holds for $G_i$. This implies that $G$ can be identified to a cluster graph with $p$ vertices. Therefore, $(G,k)$ is a \yes-instance of \kidH{Cluster}. This concludes the proof.
\end{cproof}

Because $(G,k)$ and $(G',k')$ are equivalent instances and $|V(G')|\leq p^4$, we return  $(G',k')$. Since  $M$ can be greedily selected in linear time, and the sets $R_v$ of marked vertices can also be found in linear time, the overall running time is linear. This completes the proof for \kidH{Cluster}.

\medskip
Finally, we explain kernelization for \id where $H$ is a cluster graph. As the arguments are very similar to the ones used for \kidH{Clique} and \kidH{Cluster}, we only sketch the proof and leave the details to the reader. Let $(G,H)$ be an instance of \id where $H$ is a cluster graph with $p$ vertices. Trivially, if $|V(G)|<|V(H)|$, $(G,H)$ is a \no-instance. Similarly, we have a \no-instance if the number of connected components of $G$ is less than the number of connected components of $H$ because of~\Cref{obs:connect}.
If $G$ is an edgeless graph, then $G$ admits an identification to $H$ if and only if $|V(G)|\geq |V(H)|$ and $H$ is an edgeless graph. Furthermore, if $H$ is an edgeless graph, then $G$ admits an identification to $H$ if and only if the number of connected components of $G$ is at least $|V(H)|$ by~\Cref{obs:connect}. If $(G,H)$ is one of these easy instances, we solve the problem and return either a trivial \yes or \no-instance depending on the answer. From now on, we assume that these are not the cases. 

If $H$ is a connected graph, we have that $H$ is a complete graph. Then \id for $(G,H)$ is equivalent to \kidH{Clique} for $k=n-p$. This means that we can apply the kernelization algorithm for \kidH{Clique} for the parameterization by $p=n-k$, which returns an instance $(G',k')$ of \kidH{Clique} with $|V(G')|=\Ocal(p^4)$. 
Then we return the instance $(G',K_q)$ of \id with $q=\max\{1,|V(G')|-k'\}$. Clearly, $(G,H)$ and $(G',K_q)$ are equivalent. Thus, we can assume that $H$ has at least two connected components.

Next, we apply reduction rules. It may happen that after applying one of them, we would end up in one of the cases considered above. Then we return the corresponding instance of \id. 

Suppose that $u$ is an isolated vertex of $G$ and $v$ is an isolated vertex of $H$. Then if $(G,H)$ is a \yes-instance of \id, there is an $H$-witness structure $\Wcal=\{W(x)\mid x\in V(H)\}$, where $W(v)=\{u\}$. Thus, the following rule is safe.

\begin{reduction}\label{red:isol_both_cluster}
If $G$ has an isolated vertex $u$ and $H$ has an isolated vertex $v$, set $G:=G-u$ and $H:=H-v$.    \end{reduction}

If, after the exhaustive application of \Cref{red:isol_both_cluster}, $G$ has isolated vertices, $H$ has no isolated vertices. Assume that $G$ has an isolated vertex $u$. It can be seen that the instances $(G,H)$ and $(G-u,H)$ are equivalent. This allows us to apply the next rule.

\begin{reduction}\label{red:isol_G_cluster}
If $G$ has an isolated vertex $u$, set $G:=G-u$.    
\end{reduction}

Hence, after the exhaustive application of \Cref{red:isol_G_cluster}, $G$ has no isolated vertices.  

Suppose that $G$ has at least $\binom{p}{2}$ connected components for $p=|V(H)|$. Since $G$ has no isolated vertices, each of them can be identified to a copy of $K_2$. Then we use these copies for further identifications. If $H$ has $q$ isolated vertices, we identify $q$ copies to single vertices. The other at least $\binom{p}{2}-q\geq |E(H)|$ copies of $K_2$ are used to obtain edges of $H$ by identifications in the same way as with \kidH{Clique} and \kidH{Cluster}. 
Hence, if $G$ has at least $\binom{p}{2}$ connected components for $p=|V(H)|$, we return a trivial \yes-instance. 

From now on, we assume that $G$ has at most $\binom{p}{2}-1$ connected components. Let $M$ be an inclusion-maximal matching of $G$. Suppose that there is a connected component $G'$ of $G$ containing at least $\binom{p}{2}$ edges of $M$. Consider a connected component $H'$ of $H$ whose number of vertices is maximum. Note that $H'$ is not a single-vertex graph. Furthermore, $G'$ can be identified to $H'$ in the same way as for \kidH{Clique} by making use of $\binom{|V(H')|}{2}\leq\binom{p}{2}$ edges of $M$ in $G'$. Also, we may assume that if $G$ admits an identification to $H$, then $G'$ is identified to $H'$. To see this, suppose that $G$ admits an identification to $H$. By \Cref{obs:connect}, the vertices of $G'$ are in the bags of the vertices of a connected component $H_1$ of $H$. Assume that $H_1\neq H'$, and let $\mathcal{C}_1$ and $\mathcal{C}'$ be the sets of connected components of $G$ distinct from $G'$ whose vertices are in the bags of the vertices of $H_1$ and $H'$, respectively. Then the components of $\mathcal{C}'$ can be identified to $H'$, and because $H_1$ and $H'$ are complete graphs and $|V(H_1)|\leq|V(H')|$, they can be further identified to $H_1$. Then each component of $\mathcal{C}_1$ can be identified into a single vertex, which is then identified with an arbitrary vertex of this copy of $H_1$; this does not change adjacencies, because distinct connected components of $G$ are not adjacent. Since $G'$ can be identified to $H'$, we obtain an identification of $G$ to $H$ in which $G'$ is identified to $H'$. This leads to the following reduction rule.

\begin{reduction}\label{red:comp}
If there is a connected component $G'$ of $G$ containing at least $\binom{p}{2}$ edges of an inclusion-maximal matching $M$ of $G$, then find a connected component $H'$ of $H$ with the maximum number of vertices and set $G:=G-V(G')$ and $H:=H-V(H')$.
\end{reduction}

After the exhaustive application of the rule, each connected component has a vertex cover of size at most $2(\binom{p}{2}-1)$ formed by the endpoints of the edges of $M$. Since $G$ has at most $\binom{p}{2}-1$ connected components, $G$ has a vertex cover $S$ of size at most
$2(\binom{p}{2}-1)^2$. Let $I=V(G)\setminus S$. We mark vertices of $I$ in exactly the same way as for \kidH{Cluster}. Namely, for each $v\in S$, select a set $R_v$ of at most $\binom{p}{2}$ neighbors of $v$ in $I$. Then for each connected component of $G$, we find a spanning tree $T$ and mark the non-leaf vertices of this tree in $I$. For the set of marked vertices $R$, we consider $G'=G[S\cup R]$. In exactly the same way as in~\Cref{cl:mark_cluster}, we have that the instances $(G,H)$ and $(G',H)$ are equivalent. Notice that 
$|R|\leq |S|\binom{p}{2}+|S|-1\leq 2(\binom{p}{2}-1)^2(\binom{p}{2}+1)-1$, and 
$|V(G')|\leq |S|+|R|\leq 2(\binom{p}{2}-1)^2(\binom{p}{2}+2)-1=\Ocal(p^6)$. 
Therefore, \id on instances $(G,H)$, where $H$ is a cluster graph, admits a kernel with $\Ocal(p^6)$ vertices. This concludes the proof.
\end{proof}

We remark that \id on instances $(G,H)$ where $H$ is a clique of size $p$ is equivalent to \kidH{Clique} with $k=n-p$. Also, since all vertices of complete and cluster graphs are simplicial, we have that \id on instances $(G,H)$ where $H$ is either a complete or cluster graph admits a polynomial kernel when parameterized by the number of simplicial vertices of $H$ (see~\Cref{t-2}).

\Cref{th:kernelclique} implies the following \FPT algorithms. 

\begin{corollary}\label{cor:clique-FPT}
\kidH{Clique} and \kidH{Cluster} can be solved in $p^{\Ocal(p^2)}\cdot n^{\Ocal(1)}$ time when parameterized by $p=n-k$. Furthermore, \id on instances $(G,H)$, where $H$ is a cluster graph with $p$ vertices, can be solved in $p^{\Ocal(p^2)}\cdot n^{\Ocal(1)}$ time.
\end{corollary}

\begin{proof}
First, we explain the algorithm for \kidH{Clique}.
Let $(G,k)$ be an instance of \kidH{Clique} and let $p=n-k$. After applying the kernelization algorithm from~\Cref{th:kernelclique}, we have that $n=\Ocal(p^4)$. 
Recall that $(G,k)$ is a \yes-instance if and only if $G$ can be identified to $K_p$.
If $p=n-k=1$, $(G,k)$ is a trivial \yes-instance. If $p\geq 2$ and $G$ is an edgeless graph, $(G,k)$ is a trivial \no-instance. Thus, we assume that $p\geq 2$ and $G$ has at least one edge.
Notice that $G$ can be identified to $K_p$ if and only if there is a set $A$ of $\binom{p}{2}$ edges whose endpoints can be partitioned into $p$ sets $W_1,\ldots,W_p$ such that for any distinct $i,j\in[p]$, $A$ contains an edge with one endpoint in $W_i$ and the other in $W_j$. Since $n=\Ocal(p^4)$, such a set $A$ with the corresponding partition $W_1,\ldots,W_p$ can be guessed in $p^{\Ocal(p^2)}$ time. Then the total running time is $p^{\Ocal(p^2)}\cdot n^{\Ocal(1)}$. This completes the proof for \kidH{Clique}.

For \kidH{Cluster}, the algorithm is almost the same. For an instance $(G,k)$ of \kidH{Cluster}, we apply the kernelization algorithm from \Cref{th:kernelclique}.
The problem is trivial if $p\leq 1$. Also, we can assume that $G$ has no isolated vertices. 
We use the fact that $(G,k)$ is a \yes-instance of \kidH{Cluster} if and only if $G$ can be identified to a cluster graph $H$ with $p=n-k$ vertices.
If $G$ has at least $p$ connected components, then $G$ can be identified to the cluster graph with $p$ isolated vertices. In this case, we report a \yes-instance and stop. Otherwise,  we guess the number $q< p$ of isolated vertices in a hypothetical target graph $H$. By~\Cref{obs:connect}, $G$ can be identified to such a graph only if $G$ has at least $q+1$ connected components, and we can assume that exactly $q$ connected components of $G$ are identified to these isolated vertices. We guess $q$ connected components of $G$, delete them from $G$ and set $p:=p-q$. From now on, we assume that $H$ has no isolated vertices.
Notice that $H$ has at most $\binom{p}{2}$ edges. Then we guess the set $A$ of at most $\binom{p}{2}$ edges of $G$ ensuring the adjacencies between the bags in an $H$-witness structure. In the next step, we guess the partition of the endpoints of $A$ corresponding to their inclusion in the bags.   
Finally, we verify whether the partition of the endpoints of $A$ can be extended into a partition of $V(G)$ such that the bags of distinct clusters are not adjacent. This can be done by considering the connected components of $G$ containing the endpoints of the edges of $A$.
We have $p^{\Ocal(p^2)}$ possibilities for $A$. Then, because the edges of $A$ have at most
$2\binom{p}{2}$ endpoints, we have $p^{\Ocal(p^2)}$ possible partitions of the endpoints of $A$. 
The extendability of the partition can be verified in linear time by the standard breadth-first search. Since the kernelization algorithm is polynomial, and the connected components of $G$ identified to isolated vertices can be guessed in $p^{\Ocal(p)}$ time, the total running time is
$p^{\Ocal(p^2)}\cdot n^{\Ocal(1)}$. This concludes the proof for \kidH{Cluster}.

For \id on instances $(G,H)$, where $H$ is a cluster graph, we follow the same ideas. First, we apply the kernelization algorithm from \Cref{th:kernelclique}. Assume that $(G,H)$ is  
the obtained instance. For trivial instances, we immediately solve the problem. Also, if $H$ is connected, we use the algorithm for \kidH{Clique}. Assume that these are not the cases. Notice that $G$ has no isolated vertices and not all vertices of $H$ are isolated.  If $H$ has $q$ isolated vertices, we can assume by~\Cref{obs:connect} that $q$ connected components of $G$ are identified to these $q$ isolated vertices. We guess these $q$ connected components. Then we guess the set $A$ of at most $\binom{p}{2}$ edges of $G$ ensuring the adjacencies between the bags in an $H$-witness structure. Finally, we verify whether these guesses can be extended to an $H$-witness structure. Since we consider $p^{\Ocal(p)}$ guesses of the connected components identified to the isolated vertices of $H$ and $p^{\Ocal(p^2)}$ guesses of $A$, the overall running time is $p^{\Ocal(p^2)}\cdot n^{\Ocal(1)}$. This completes the proof.
\end{proof}

\section{Identification to Split, Interval and Chordal Graphs}\label{sec:chord}
In this section, we consider \Hid{$\Hcal$}, where $\Hcal$ is the class of split, interval, or chordal graphs.

\subsection{Algorithmic Lower Bounds}
First, we show algorithmic and kernelization lower bounds using the results for \kidH{Clique}. 

\begin{theorem}\label{th:NPsplit}
\Hid{$\Hcal$}, where $\Hcal$ is the class of split, interval, or chordal graphs,
is \NP-complete and does not admit a polynomial kernel
when parameterized by $k$ unless \coNP$\subseteq$~\NPp.
\end{theorem}

\begin{proof}
In~\Cref{th:NPclique}, we proved that \kidH{Clique} is \NP-complete and does not admit a polynomial kernel when parameterized by $k$ unless \coNP$\subseteq$~\NPp. We use this and prove both claims of the theorem by presenting a polynomial parameter transformation from \kidH{Clique} parameterized by $k$.

Let $(G,k)$ be an instance of \kidH{Clique}. We assume that $n\geq k+1$ as, otherwise, 
$(G,k)$ is a trivial \yes-instance. We construct $G'$ from $G$ by adding an independent set $I$ of size $2k+2$, and adding an edge between each vertex of $I$ and each vertex of $V(G)$. 

If $(G,k)$ is a \yes-instance of \kidH{Clique} with an $H$-witness structure $\Wcal=\{W(x)\mid x\in V(H)\}$ for some clique $H$, then $\Wcal'=\Wcal\cup\{\{v\}\mid v\in I\}$ is an $H'$-witness structure in $G'$, where $H'$ is the graph obtained from the clique $H$ by adding an independent set $I'$ composed of $2k+2$ vertices and making these vertices adjacent to each vertex of $H$. Since $H$ is a clique, 
$H'$ is both a split and an interval graph, and therefore also a chordal graph.
Because the bags of $\Wcal'$ corresponding to the vertices $x\in I'$ are trivial, we need the same number of identifications
to identify $G'$ to $H'$ as we need for identifying $G$ to $H$.
Therefore, $(G',k)$ is a \yes-instance of \Hid{$\Hcal$}.

Suppose now that $(G',k)$ is a \yes-instance of \Hid{$\Hcal$} with an $H$-witness structure 
$\{W(x)\mid x\in V(H)\}$ for some graph $H\in\Hcal$. Let $X=\{x\in V(H)\mid W(x)\subseteq V(G)\}$. We claim that $H[X]$ is a non-empty clique.

Recall that $|V(G)|\geq k+1$. Then there is $x\in V(H)$ such that $W(x)\subseteq V(G)$. Indeed, if for every $x\in V(H)$, $W(x)\cap I\neq\emptyset$, then we need at least $k+1$ identifications to obtain $H$ from $G'$. Thus, $X\neq\emptyset$.  
To show that $H[X]$ is a clique, note that because $|I|=2k+2$, we know that at least two vertices of $I$ are not identified, i.e., there exist $a,b\in I$ and $x_a,x_b\in V(H)$ such that $W(x_a)=\{a\}$ and $W(x_b)=\{b\}$.
For a contradiction, assume that there exist vertices $x,y\in X$ with non-adjacent $W(x)$ and $W(y)$. In that case, $\{x_a,x,x_b,y\}$ induces a cycle on four vertices in $H$, contradicting the fact that $H$ is chordal. Thus, $H[X]$ is a non-empty clique.

Consider $Y=\{x\in V(H)\colon W(x)\cap V(G)\neq\emptyset\}$. If $X=Y$, then $\{W(x)\mid x\in X\}$ is an $H[X]$-witness structure for $G$ demonstrating that $(G,k)$ is a \yes-instance of \kidH{Clique}. Suppose that $Y\setminus X\neq\emptyset$. Let $z\in X$ be an arbitrary vertex of the non-empty $X$.
Consider
$\Wcal'=\{W'(x)\mid x\in X\}$, where $W'(z)=W(z)\cup\bigcup_{y\in Y\setminus X}(W(y)\cap V(G))$ and $W'(x)=W(x)$ for $x\in X$ such that $x\neq z$. Because $H[X]$ is a clique, we have that $\Wcal'$ is a clique-witness structure. Since for every $y\in Y\setminus X$, 
$|W(y)\cap V(G)|<|W(y)|$, 
$\sum_{x\in X}(|W'(x)|-1)\leq \sum_{y\in Y}(|W(y)|-1)\leq k$. 
Thus, $(G,k)$ is a \yes-instance of \kidH{Clique}.
This completes the proof.
\end{proof}

In~\Cref{thm:forest_XP}, we proved that \id can be solved by an \XP-algorithm on instances $(G,F)$ where $F$ is a forest with $\ell$ leaves, when parameterized by $\ell$. Notice that leaves are simplicial vertices of a forest. Here, we show that the result for forests cannot be generalized for chordal graphs parameterized by the number of simplicial vertices---the problem is \para-complete when parameterized by the number of simplicial vertices.

\begin{theorem}\label{th:NPchordal}
\id is \NP-complete on instances $(G,H)$ where $H$ is both a split and an interval graph with two simplicial vertices.
\end{theorem}

\begin{proof}
We present a reduction from \kidH{Clique}, which is \NP-complete by \Cref{th:NPclique}.
Let $(G,k)$ be an instance of \kidH{Clique}. We assume that $n=|V(G)|\geq k+2$, as otherwise, we have a trivial \yes-instance.
We construct $G'$ and $H$ as follows. To construct $G'$, we 
\begin{itemize}
\item construct a copy of $G$,
\item add two vertices $a$ and $b$ adjacent to all vertices of $G$,
\item construct a vertex $a'$ and make it adjacent to $a$.
\end{itemize}
For constructing $H$, we 
\begin{itemize}
\item construct a clique $Q$ of size $n-k$,
\item add two vertices $x_a$ and $x_b$ adjacent to all vertices of $Q$,
\item construct a vertex $x_a'$ and make it adjacent to $x_a$.
\end{itemize}
See \Cref{fig:NPchordal} for an illustration.
Notice that $H$ is both split and interval, and $x_a'$ and $x_b$ are the only simplicial vertices of $H$.

\begin{figure}[ht]
    \centering
    \includegraphics[scale=1]{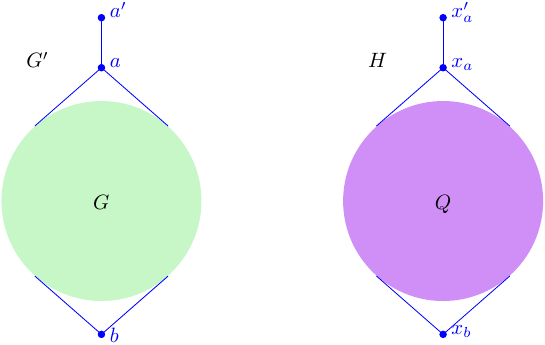}
    \caption{Illustration of the graph $G'$ and the graph $H$ in the proof of \Cref{th:NPchordal}.}
    \label{fig:NPchordal}
\end{figure}

We prove that $(G,k)$ is a \yes-instance of \kidH{Clique} if and only if $(G',H)$ is a \yes-instance of \id. As the class of cliques is closed under identifications, $(G,k)$ is a \yes-instance if and only if $G$ can be identified to a clique of size $n-k$. Therefore, we have to show that $G$ can be identified to a clique of size $n-k$ if and only if $G'$ can be identified to $H$.

Suppose that $G$ can be identified to a clique of size $n-k$. Because $Q$ in the construction of $H$ is of size $n-k$, $G$ can be identified to $H[Q]$ with an $H[Q]$-witness structure $\Wcal=\{W(x)\mid x\in Q\}$. We extend the witness structure on $H$ by defining
$W(x_a)=\{a\}$, $W(x_a')=\{a'\}$, and $W(x_b)=\{b\}$. By straightforwardly checking the adjacency between the bags, we conclude that we have an $H$-witness structure for $G'$.

For the opposite implication, assume that $G'$ can be identified to $H$ with an $H$-witness structure $\Wcal=\{W(x)\mid x\in V(H)\}$. 
Observe that $a',b$ is the unique pair of vertices at distance three in $G'$. 
Similarly, $x_a',x_b$ is the unique pair of vertices at distance three in $H$. This implies that either $W(x_a')=\{a'\}$ and $W(x_b)=\{b\}$ or, symmetrically, $W(x_a')=\{b\}$ and $W(x_b)=\{a'\}$. 
However, in the second case, we have that the vertices of $G$ are in $W(x_a)$ because they are adjacent to $b$. 
This contradicts the fact that $\Wcal$ is an $H$-witness structure since $|Q|=n-k\geq 2$. Thus, $W(x_a')=\{a'\}$ and $W(x_b)=\{b\}$. 
This implies that $a\in W(x_a)$. Since $a$ is the unique vertex of $G'$ at distance two from $b$, and $x_a$ is the unique vertex of $H$ at distance two from $x_b$, we obtain that $W(x_a)=\{a\}$. Then $\Wcal'=\{W(x)\mid x\in Q\}$ is an $H[Q]$-witness structure. 
This means that $G$ can be identified to a clique of size $n-k$ and concludes the proof.
\end{proof}

\subsection{FPT Algorithm for Identification to Split Graphs}
In~\Cref{th:NPsplit}, we proved that \Hid{$\Hcal$}, where $\Hcal$ is the class of split graphs, does not admit a polynomial kernel
when parameterized by $k$ unless \coNP$\subseteq$~\NPp. Here, we prove that the problem is \FPT for this parameterization and refer to the problem as \kidH{Split}.

Given an $H$-witness structure $\Wcal=\{W(x)\colon x\in V(H)\}$ for $G$,  
where $H$ is a split graph with a given partition of the vertex set into a clique and an independent set, we say that $W(x)$ is an \emph{independent-bag} if $x$ is in the independent set,
and $W(x)$ is a \emph{clique-bag} if $x$ is in the clique. We also say that $\Wcal$ is \emph{nice} if $|W(x)|=1$ for all independent-bags. We show that a nice witness structure always exists.

\begin{lemma}\label{lem:singletons_split}
Let $G$ be a graph that can be identified to a split graph by at most $k$ identifications.
Then $G$ can be identified, by at most $k$ identifications, to a split graph $H=(K\cup I,E)$, where $H[K]$ is a clique and $I$ is an independent set, with a nice $H$-witness structure $\Wcal=\{W(x)\mid x\in V(H)\}$.
\end{lemma}

\begin{proof} 
The proof follows the same arguments as the proof of~\Cref{lem:singletons}.
Let $H=(K\cup I,E)$ be a split graph such that $G$ can be identified to $H$ by at most $k$ identifications, with an $H$-witness structure $\Wcal=\{W(x)\mid x\in V(H)\}$ where the
total number of vertices of $G$ in the bags $W(x)$ for  $x\in I$ is minimum. 
By contradiction, we prove that $|W(x)|=1$ for all $x\in I$.

Suppose that there is $y\in I$ such that $|W(y)|\geq 2$. Consider an arbitrary  
$v\in W(y)$ and an arbitrary $z\in K$.
We define $\Wcal'=\{W'(x)\mid x\in V(H)\}$, where $W'(y)=\{v\}$, $W'(z)=W(z)\cup (W(y)\setminus \{v\})$, and $W'(x)=W(x)$ for all $x\in V(H)$ distinct from $y$ and $z$. Let $H'$ be the graph with $V(H')=K\cup I$ such that two distinct vertices $x,x'$ are adjacent if and only if $W'(x)$ and $W'(x')$ are adjacent, that is, $\Wcal'$ is an $H'$-witness structure. 
For distinct vertices $x,x'\in K$, $xx'\in E(H')$ because $y\neq x,x'$. For every $x\in I$ distinct from $y$, $N_H(x)=N_{H'}(x)$ as all the neighbors of $x$ in $H$ are in $K$. Therefore, $H'$ is a split graph, where $H'[K]$ is the clique. 
Since $\sum_{x\in K\cup I}(|W'(x)|-1)=\sum_{x\in K\cup I}(|W(x)|-1)\leq k$, $H'$ can be obtained from $G$ by at most $k$ identifications. 
However, the total number of vertices in $\bigcup_{x\in I}W'(x)$ is less than the size of $\bigcup_{x\in I}W(x)$. This contradicts the choice of $\Wcal$ and proves the claim.
\end{proof}

In our algorithm for \kidH{Split}, we will reduce the problem to \kidH{Clique}. For this, we need the following technical lemma.

\begin{lemma}\label{lem:reduction_clique}
Let $G$ be a graph whose vertex set has a partition into two sets $U\neq\emptyset$ and $I$, where $I$ is an independent set which may be empty. Then $G$ can be identified to a split graph $H$ by at most $k$ identifications admitting a nice $H$-witness structure $\Wcal=\{W(x)\mid x\in V(H)\}$ such that every vertex of $U$ is in a clique-bag if and only if $G[U]$ can be identified to a clique by at most $k$ identifications.
\end{lemma}

\begin{proof}
For the forward implication,
consider a split graph $H$ with at least $n-k$ vertices having a nice $H$-witness structure $\Wcal=\{W(x)\mid x\in V(H)\}$ for $G$ such that every vertex of $U$ is in a clique-bag and the number of independent-bags is maximum. 
Let $K$ be the vertex set inducing the corresponding clique of $H$. 
Consider $K'=\{x\in V(H)\mid W(x)\cap U\neq\emptyset\}\subseteq K$. 
We claim that $\Wcal_U=\{W(x)\mid x\in K'\}$ is a partition of $U$.

Clearly, the bags $W(x)$ of $\Wcal$ are disjoint, and by definition, $U\subseteq\bigcup_{x\in K'}W(x)$. Thus, we have to show that $I\cap\big(\bigcup_{x\in K'}W(x)\big)=\emptyset$.

If $|K'|=1$, then for the unique $x\in K'$, $U\subseteq W(x)$. Then because the number of 
independent-bags is maximum, we immediately obtain that $K=K'$ and the vertices of $I$ are in singleton-bags. Therefore, the claim holds. Suppose that $|K'|\geq 2$. 
For the sake of contradiction, assume that there is $y\in K'$ such that $W(y)\cap I\neq\emptyset$. Let $\{v_1,\ldots,v_p\}=W(y)\cap I$ and consider an arbitrary $z\in K'$ distinct from $y$.
We construct the
graph $H'$ and an $H'$-witness structure $\Wcal'=\{W'(x)\mid x\in V(H')\}$ for $G$ from $H$ as follows:
\begin{itemize}
\item delete $y$ and introduce new vertices $r_1,\ldots,r_p$,
\item set $W'(z)=W(z)\cup(W(y)\cap U)$,
\item for $i\in[p]$, set $W'(r_i)=\{v_i\}$,
\item for all distinct $x_1,x_2\in V(H')$, make $x_1$ and $x_2$ adjacent in $H'$ if and only if $W'(x_1)$ and $W'(x_2)$ are adjacent,
\item set $W'(x)=W(x)$ for all $x\in V(H)\setminus\{y,z\}$.
\end{itemize}
Notice that because $W'(r_1),\ldots,W'(r_p)$ are singletons, we have
\[\sum_{x\in V(H')}(|W'(x)|-1)=\sum_{x\in V(H)\setminus\{y,z\}}(|W(x)|-1)+(|W(z)|-1)+(|W(y)|-p)\leq \sum_{x\in V(H)}(|W(x)|-1).\] Thus, $H'$ can be obtained from $G$ by at most $k$ identifications. 

Observe that if $x_1$ and $x_2$ are distinct vertices of $K\setminus \{y\}$, then $W'(x_1)\supseteq W(x_1)$ and $W'(x_2)\supseteq W(x_2)$ are adjacent because $W(x_1)$ and $W(x_2)$ are adjacent. Therefore, $K\setminus\{y\}$ is a clique of $H'$. 
Recall that the vertices of $I$ form an independent set and their neighbors are in $U$. Thus, $v_1,\ldots,v_p$ have their neighbors in the bags $W'(x)$ for $x\in K'\setminus\{y\}$. 
Thus, $(V(H)\setminus K)\cup\{r_1,\ldots,r_p\}$ is an independent set of $H'$.
Therefore, $H'$ is a split graph, and $\Wcal'$ is a nice $H'$-witness structure. However, this contradicts the choice of $H$ and $\Wcal$. Thus, there is no $y\in K'$ with $W(y)\cap I\neq\emptyset$. 

This proves that $\Wcal_U=\{W(x)\mid x\in K'\}$ is a partition of $U$. Since $\sum_{x\in K'}(|W(x)|-1)\leq \sum_{x\in V(H)}(|W(x)|-1)\leq k$, we have that
$G[U]$ can be identified to a clique by at most $k$ identifications. 

The opposite implication is straightforward. If $G[U]$ can be identified to a clique by at most $k$ identifications, then these identifications applied to $G$ result in a split graph with the independent set $I$.
This completes the proof.
\end{proof}

\begin{theorem}\label{th:FPTsplit}
\kidH{Split} is solvable in $k^{\Ocal(k)}\cdot n^{\Ocal(1)}$ time.
\end{theorem}

\begin{proof}
Let $(G,k)$ be an instance of \kidH{Split}. We assume that $G$ is not a split graph, as otherwise, the problem is trivial. 

Suppose that $(G,k)$ is a \yes-instance of \kidH{Split}. Then, because the class of split graphs is closed under vertex deletions, we have that there is a set of size at most $2k$ vertices whose deletion results in a split graph---we delete at most $2k$ vertices involved in the identification. We use the fact that, in $\Ocal(1.2738^hh^{\Ocal(\log h)}+n^3)$ time~\cite{CyganP13spli}, it can be decided whether it is possible to delete at most $h$ vertices to obtain a split graph. We first use the algorithm from~\cite{CyganP13spli} to find a set $S\subseteq V(G)$ of size at most $2k$ such that $G-S$ is a split graph. If such a set does not exist, we conclude that $(G,k)$ is a \no-instance. From now on, we assume that $S$ is given. Furthermore, we suppose that $S$ is an inclusion-minimal set whose deletion results in a split graph.

Let $G'=G-S$, and consider a partition of $V(G')$ into a clique $G'[K]$ and an independent set $I$ ($I$ may be empty). Since $G$ is not a split graph and $S$ is an inclusion-minimal set whose deletion results in a split graph, we can choose $K$ to have size at least~$2$.

By~\Cref{lem:singletons_split}, we can assume that the independent-bags in a split-witness structure are singletons. Thus, at most one vertex $w\in K$ can be in an independent-bag. We guess this vertex by considering all possibilities, and set $K:=K\setminus\{w\}$ and $S:=S\cup\{w\}$ if such a vertex $w$ is selected; the sets are not changed if, in the considered choice, the vertices of $K$ are in clique-bags. From now on, we assume that $K$ and $S$ are fixed.

Notice that in the considered $K$, every vertex should be in a clique-bag in a hypothetical solution. Vertices of $S$ may be either in clique-bags or independent-bags. Also, by~\Cref{lem:singletons_split}, we can assume that the vertices of $S$ from independent-bags are the unique vertices of these bags and their neighbors in $G$ are in clique-bags.   
We consider all possible partitions of $S$ into $S_K$ and $S_I$ (one of the sets may be empty), where $S_I$ is an independent set, to guess the partition of $S$ into the set of vertices $S_K$ in the clique-bags and the set of vertices $S_I$ in independent-bags.
For the considered choice of $S_I$, the vertices of $N_G(S_I)$ should be in clique-bags. We define $U=K\cup S_K\cup N_G(S_I)$ and set $I'=(I\setminus N_G(S_I))\cup S_I$. Then $U$ and $I'$ form a partition of $V(G)$, where $I'$ is an independent set. Furthermore, if $(G,k)$ is a \yes-instance of \kidH{Split}, for one of the considered choices, the vertices of $U$ should be in clique-bags. Then, by~\Cref{lem:reduction_clique}, $(G[U],k)$ should be a \yes-instance of \kidH{Clique}. We verify this by making use of the algorithm from~\Cref{th:FPTclique}. If the algorithm reports that $(G[U],k)$ is a \yes-instance of \kidH{Clique}, then by~\Cref{lem:reduction_clique}, we conclude that $(G,k)$ is a \yes-instance of \kidH{Split}.
If we fail to find such an instance for all choices of $w$ and all partitions of $S$, we have that $(G,k)$ is a \no-instance of \kidH{Split}. This completes the description of the algorithm and its correctness proof.

For the running-time analysis, recall that $S$ (if it exists) can be found in $2^{\Ocal(k)}\cdot n^{\Ocal(1)}$ time. Then we have at most $n+1$ choices of $w$ (including the choice of no $w$) and, as $|S|\leq 2k+1$, at most $2^{2k+1}$ partitions of $S$. Finally, the algorithm from~\Cref{th:FPTclique} works in $k^{\Ocal(k)}\cdot n^{\Ocal(1)}$ time. Then the overall running time is $k^{\Ocal(k)}\cdot n^{\Ocal(1)}$. This concludes the proof.
\end{proof}

\subsection{Kernelization for Dual Parameterization}
We conclude this section by proving that for the parameterization by $p=n-k$, \Hid{$\Hcal$} admits a trivial polynomial kernel, where $\Hcal$ is the class of split, interval, or chordal graphs. 

\begin{theorem}\label{thm:chordal_kernel}
Let $\Hcal$ be the class of split, interval, or chordal graphs (or more generally, any class containing the complete graphs, edgeless graphs, stars, and the disjoint unions of a star and isolated vertices).
Let $G$ be a graph and $k\in\bN$. Let $p=n-k\geq 1$.
If $n\ge p^2-3$,
then $(G,k)$ is a \yes-instance of \Hid{$\Hcal$}. In particular, \Hid{$\Hcal$} admits a kernel with at most $p^2$ vertices.
\end{theorem}

\begin{proof}
Let $(G,k)$ be an instance of \Hid{$\Hcal$}. We construct the following kernelization algorithm.
If $G$ is an edgeless graph, then $G\in\Hcal$ and $(G,k)$ is a \yes-instance. Assume that this is not the case.  

If $G$ has a matching of size at least $\binom{p}{2}$, then $G$ is a \yes-instance of \kidH{Clique} by the arguments from~\Cref{cl:matching}. Recall that in this case, we can construct a $K_p$-witness structure by putting the endpoints of $\binom{p}{2}$ edges of the matching in the bags to ensure that the bags are pairwise adjacent, and other vertices can be placed in the bags arbitrarily. Since $K_p\in\Hcal$, we have that $(G,k)$ is a \yes-instance of \Hid{$\Hcal$}.  
Otherwise, $G$ has a vertex cover $S$ of size at most $2\cdot\binom{p}{2}-2$ formed by the endpoints of a maximum matching. Since $G$ has at least one edge, $S\neq\emptyset$. 
If $|V(G-S)|\ge p-1$, then we can identify $S$ into a single vertex. We thus obtain either a star with at least $p$ vertices or the disjoint union of a star and some isolated vertices with at least $p$ vertices. This implies that $G$ is a \yes-instance of \Hid{$\Hcal$}. 

If we concluded that $(G,k)$ is a \yes-instance, then the kernelization algorithm returns the trivial \yes-instance $(G=(\{v\},\emptyset),0)$, where $G$ has a single vertex. Otherwise, we have $|S|\leq2\binom{p}{2}-2$ and $|V(G)\setminus S|\leq p-2$. Hence, $|V(G)|\leq 2\binom{p}{2}-2+p-2=p^2-4$. The algorithm then returns the input instance $(G,k)$. This concludes the proof.
\end{proof}

Combining the kernelization from \Cref{thm:chordal_kernel} with the brute force guessing of at most $k\leq p^2$ identifications, we obtain the following corollary.

\begin{corollary}\label{cor:chordal_FPT}
\Hid{$\Hcal$}, where $\Hcal$ is the class of split, interval, or chordal graphs, can be solved in $p^{\Ocal(p^2)}\cdot n^{\Ocal(1)}$ time when parameterized by $p=n-k$.
\end{corollary}

\section{Future Work}\label{s-con}\label{sec:ccl}

We presented a comprehensive study of the complexity of {\sc ${\cal H}$-Identification} and \id for ten well-known classes of chordal graphs, including the class of chordal graphs themselves. Our results are summarized in Tables~\ref{t-1} and~\ref{t-2}. By using a systematic approach, we were able to identify two open problems: 
Is \Hid{$\cal H$}, parameterized by $k$, \FPT\ or \W{1}-hard when ${\cal H}$ is the class of interval graphs or the class of chordal graphs?

It is also possible to consider alternative parameters. For example, we may ask whether a given graph $G$ can be modified into a graph from a specified class ${\cal H}$ by involving at most $p$ vertices in a sequence of identifications. This variant was studied in \cite{MorelleST26vert}, where a polynomial kernel was obtained when 
${\cal H}$ is the class of forests (through a reduction to {\sc Vertex Cover}). Since $k$ identifications involve $p\leq 2k$ vertices,
this result yields an \FPT\ $2$-approximation algorithm for \kidH{Forest}.

Finally, we recall that for the minor relation, the set of permitted graph operations consists of vertex deletion, edge deletion, and edge contraction. If we only allow vertex deletion and edge contraction, then we obtain the {\it induced minor} relation. In both relations, we could replace edge contraction by vertex identification.  In the first case, we obtain the {\it identification minor} relation introduced in~\cite{MorelleST26vert}, and in the second case, we have an {\it induced identification minor} relation.

We note that the corresponding decision problems
$H$-{\sc Identification Minor} and $H$-{\sc Induced Identification Minor} are in \XP, when parameterized by $|V(H)|$. This can be seen as follows. For each edge $xy\in E(H)$, we guess an edge $uv\in E(G)$ where we place $u$ in $W(x)$ and $v$ in $W(y)$. We also guess a vertex $t$ of $G$ to be in a bag $W(z)$ for each isolated vertex~$z\in V(H)$.
We assign the set $R$ of all non-guessed vertices to an arbitrary bag. For the variant $H$-{\sc Identification Minor}, we delete all edges incident to $R$ and between vertices of bags that must be non-adjacent. 
For the variant $H$-{\sc Induced Identification Minor}, we simply delete all vertices of $R$. We now identify each bag to a single vertex. If we do not obtain $H$ (which may happen when edge deletions are not allowed), then we discard the branch and try again. 

We leave a more extensive complexity study of 
${\cal H}$-{\sc Identification Minor} and ${\cal H}$-{\sc Induced Identification Minor} for future work.
As a final remark, we note that the identification minor relation is  well-quasi-ordered~\cite{MorelleST26vert}, which implies that every identification-minor-closed graph class admits a finite obstruction set. In contrast, this does not hold for the induced identification minor relation: for example, the complements of cycles form an infinite antichain (see also~\cite{MorelleST26vert} for another infinite antichain).  

\section*{AI disclosure}
The results were obtained and written without use of AI tools. ChatGPT Astra was used in the final stage of the manuscript preparation to identify minor mistakes and inconsistencies and to correct grammatical errors. The authors take full responsibility for the correctness of the paper.

\end{document}